\documentclass[lettersize,journal]{IEEEtran}
\usepackage{amsmath,amsfonts}
\usepackage{algorithmic}
\usepackage{algorithm}
\usepackage{array}
\usepackage[caption=false,font=normalsize,labelfont=sf,textfont=sf]{subfig}
\usepackage{textcomp}
\usepackage{stfloats}
\usepackage{url}
\usepackage{verbatim}
\usepackage{graphicx}
\usepackage{cite}

\usepackage{xcolor}
\usepackage{amsthm}
\usepackage{amssymb}
\usepackage{amsmath}
\newtheorem{proposition}{Proposition}

\begin{document}

\title{Flexible Bit-Truncation Memory for Approximate Applications on the Edge}

\author{William~Oswald, Mario~Renteria-Pinon,~\IEEEmembership{Member,~IEEE,} Md.~Sajjad~Hossain, Kyle~Mooney,~\IEEEmembership{Graduate Student~Member,~IEEE}, Md.~Bipul~Hossain, Destinie~Diggs, Yiwen~Xu, Mohamed~Shaban,~\IEEEmembership{Senior Member,~IEEE,} Jinhui~Wang,~\IEEEmembership{Senior Member,~IEEE,} and~Na~Gong,~\IEEEmembership{Senior Member,~IEEE.}

\thanks{This work has been submitted to the IEEE for possible publication.
Copyright may be transferred without notice, after which this version may
no longer be accessible.}

\thanks{Manuscript received October 15, 2025. W. Oswald and M. Renteria-Pinon contributed equally to this paper. (Corresponding author: N. Gong.)}

\thanks{W. Oswald, M. B. Hossain and M. Shaban are with the Department of Electrical and Computer Engineering, University of South Alabama, AL 36688 USA (e-mail: (\{wdo1621, mh2238 \}@jagmail.southalabama.edu, \{mshaban\}@southalabama.edu)).}

\thanks{M. Renteria-Pinon is with the Department of Electrical and Computer Engineering, New Mexico State University, NM 88003 USA (e-mail: marior3@nmsu.edu).}

\thanks{M. S. Hossain, K. Mooney, D. Diggs, J. Wang, and N. Gong are with the Department of Electrical and Computer Engineering, University of Alabama, AL 35487 USA (e-mail: (\{mhossain80, kamooney3, dadiggs1\}@crimson.ua.edu, \{jwang231, ngong\}@ua.edu)).}

}



\maketitle

\begin{abstract} Bit truncation has demonstrated great potential to enable run-time quality-power adaptive data storage, thereby optimizing the power/energy efficiency of approximate applications and supporting their deployment in edge environments. However, existing bit-truncation memories require custom designs for a specific application. In this paper, we present a novel bit-truncation memory with full adaptation flexibility, which can truncate any number of data bits at run time to meet different quality and power trade-off requirements for various approximate applications. The developed bit-truncation memory has been applied to two representative data-intensive approximate applications: video processing and deep learning. Our experiments show that the proposed memory can support three different video applications (including luminance-aware, content-aware, and region-of-interest-aware) with enhanced power efficiency (up to 47.02\% power savings) as compared to state-of-the-art. In addition, the proposed memory achieves significant (up to 51.69\%) power savings for both baseline and pruned lightweight deep learning models, respectively, with a low implementation cost (2.89\% silicon area overhead).
\end{abstract}


\section{Introduction}
Today, the demand for designing efficient computing systems is increasing due to the diminishing benefits of semiconductor technology scaling and the growing requirements of data-intensive applications, such as video processing and deep neural networks (DNNs), on resource-constrained edge devices \cite{roy2015approximate}. Fortunately, many applications have approximate, non-deterministic specifications with multiple acceptable quality levels for different conditions, providing exciting quality-adaptive design opportunities for those approximate applications \cite{alioto2018energy}. Quality-adaptive designs dynamically adjust output quality to optimize power efficiency while satisfying the quality requirements, thereby enabling the deployment of data-intensive applications on edge devices such as video steaming \cite{haidous2022content, cao2018video} or Edge AI \cite{alioto2018energy, young2025low}.

To enable quality-adaptive computing systems, bit truncation is one of the most widely-applied techniques for different approximate applications. For example, in video systems, i.e., luminance-aware \cite{chen2015vcas, edstrom2016luminance,chen2018viewer}, content-aware \cite{edstrom2019content}, and ROI-aware \cite{ROI-AwareVideoStorage}, one custom bit-truncation memory was developed to enable the specific number of truncated bits, such as three or four bits truncation in \cite{chen2015vcas, edstrom2016luminance,chen2018viewer}, zero to four bits truncation in \cite{edstrom2019content}, or zero or three bits in \cite{ROI-AwareVideoStorage}. Therefore, existing bit-truncation memories cannot provide flexibility to support different video applications. In addition, state-of-the-art bit-truncation memory designs, which have been developed for a specific video application, cannot be used for other approximate applications such as DNNs. 

In this paper, we propose a truly flexible bit-truncation memory, coined as TrunMem. Compared to existing work, TrunMem can enable run-time quality adaptation to meet the requirements of different applications. For example, with full flexibility, TrunMem can be used to support all three video systems (i.e., luminance-aware\cite{chen2015vcas, edstrom2016luminance,chen2018viewer}, content-aware \cite{edstrom2019content}, and ROI-aware video storage \cite{ROI-AwareVideoStorage}). Earlier in \cite{oswald2024flexible}, we presented a basic TrunMem design, including some preliminary results. We extend our original work with the following key contributions:

\begin{itemize}

   \item \textbf {Complete design of TrunMEM with control unit}: From a hardware design perspective, a complete circuit diagram of TrunMEM is presented, and a detailed operation process with all control units for truncation, byte mode, and word mode is provided to support different approximate applications (Section III.A).
  \item \textbf {Full-chip implementation}: A full-chip design of TrunMEM is presented and discussed; based on it, the silicon area cost is evaluated and analyzed (Section IV.B).
  \item \textbf {Thorough post-layout evaluation}: Instead of schematics-based pre-simulations in \cite{oswald2024flexible}, we incorporated the extracted parasitic parameters and performed a comprehensive suite of post-layout simulations on TrunMEM to evaluate speed, performance overhead, and power efficiency(Section IV.B). 
  \item \textbf {Software-hardware co-design framework for edge intelligence}: Based on TrunMEM, we present a software-hardware co-design framework to accelerate AI deployment in dynamic edge environments. Specifically, the software-level model compression techniques are used to train lightweight models; based on those optimized models, TrunMEM can further enable run-time adaption in the model inference process to meet the dynamic requirements of the edge applications (Section II.B).
  \item \textbf {Thorough DNN evaluation}: Two widely-used deep learning applications, including classification and object detection, are used to validate the TrunMEM-based software-hardware co-design framework. For each deep learning application, We applied TrunMEM to both baseline and pruned lightweight models with different model architectures and datasets. It can be concluded that integrating software-level compression techniques into the model training process and TrunMEM-based bit truncation into the inference process can enable optimized AI deployment on edge devices (Section IV.D).
  \item \textbf {Mathematical models for the optimal truncation value in DNNs}: Novel mathematical models are developed to identify the dummy value setting for the truncated bits for deep learning such that the expected mean-squared error of weights is minimized (Section III.A).
  \item \textbf {Comprehensive video analysis}: The video analysis in \cite{oswald2024flexible} was based on a single video. In order to evaluate the effectiveness of TrunMEM in different video applications, 4,525 videos with a variety of video characteristics, are used for luminance-aware, content-aware, and ROI-aware video storage in this paper. Also, multiple widely-used video metrics such as peak signal-to-noise ratio (PSNR), structural similarity (SSIM), and visual quality, as well as statistical analysis are included in the analysis (Section IV.C).
  \item \textbf {Open Source}: The code for all developed emulators, including those for videos and DNN applications, is open-source on GitHub."
 \end{itemize} 

 It should be emphasized that although this paper applies and tests the proposed TrunMem for video and DNN systems, which are two representative approximate data-intensive applications, TrunMEM is also applicable for different approximate applications such as audio processing, wireless applications, and Recognition, Mining and Synthesis (RMS) \cite{roy2015approximate}. In addition, this paper focuses on SRAM, the mainstream embedded memory technology; however, other memory technologies, such as DRAM or emerging memories \cite{gong2025ai}, can also benefit significantly from the TrunMem architectures, enabling run-time quality-power adaptation.
 
 The rest of the paper is organized as follows. Section II presents background and the state-of-the-art. The proposed TrunMem is detailed in Section III. The evaluation results are provided in Section IV. Finally, the conclusion is drawn in Section V.

\section{Background and State-Of-the Art}
\subsection{Bit Truncation and Approximate Video Storage} 

Compared to other techniques used for quality-adaptation, e.g., dynamic voltage scaling, bit truncation has two major advantages: (i) it can enable more energy savings \cite{frustaci2016approximate}. With bit truncation, researchers developed several viewer-aware video memories to optimize power savings, as shown in Fig. \ref{fig:video}. For example, studies applied bit truncation to luminance-aware video memory design and revealed that more least-significant-bits (LSBs) of pixel data can be truncated if the video device is operating under high lighting conditions \cite{chen2015vcas, edstrom2016luminance,chen2018viewer}. Specifically, Edstrom et al. have shown three and four LSBs per-byte can be truncated in video streams exposed to overcast and sunlight. The power-quality trade-off considers the impact of video content and viewer experience \cite{edstrom2019content}. The circuit truncates the number of LSBs between zero-bit to 4-bit according to the micro-block characteristics of videos (i.e., the average plain macroblock percentage of a video). Recently, Haidous et al. \cite{ROI-AwareVideoStorage} presented a Region-of-Interest (ROI)-aware video memory. The system truncates three LSBs for non-ROI regions of videos to further optimize video quality-power trade-off. These works use bit-truncation memory failures in high noise-tolerance viewing contexts, by adaptively disabling the LSBs of the video data stored in memories. Unfortunately, for each of the above bit-truncation techniques, custom bit-truncation memory has been developed to meet the need of the truncated bits and adaptation. In this paper, a new memory - TrunMEM is presented to support a general hardware platform for bit truncation and the detailed analysis is presented in Section IV. 

\begin{figure}
\centering
  \includegraphics[width=\linewidth]{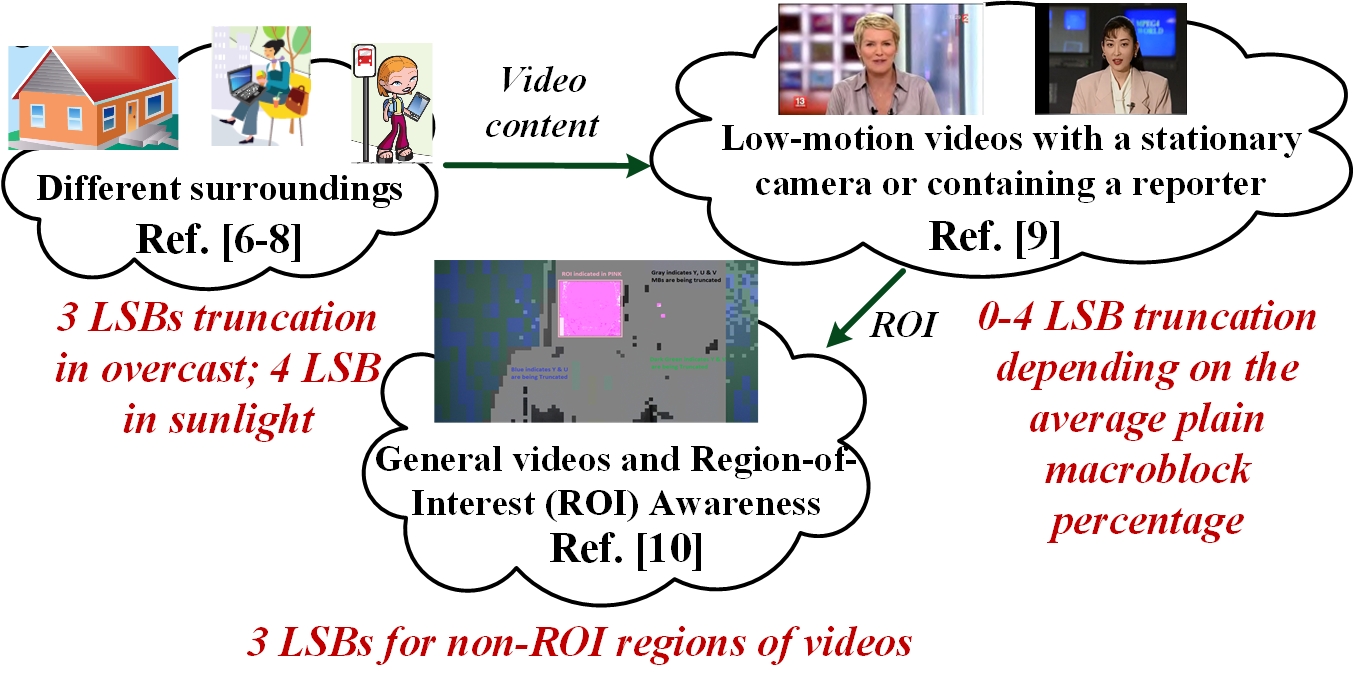}
  \vspace{-0.2in}
  \caption{Bit truncation for viewer-aware video storage.}
  \label{fig:video}
\end{figure}


\subsection{Bit Truncation and Approximate Deep Learning}

  In addition to video applications, TrunMEM is also applicable to DNN applications to enable software-hardware co-design for edge intelligence. 
  
  As the application of DNNs become more widespread on edge devices, the necessity for DNN power optimization has become essential. Recently, several software-based model compression approaches have been developed to design lightweight AI models for edge devices such as quantization, knowledge distillation, and pruning \cite{model_compression_forDNN_survey, computational_complexity_rediction}. These software approaches can substantially reduce the size and computational cost of the model. 
  
  Quantization reduces the numerical precision of model weights and activations, thereby lowering memory requirements and enabling deployment on devices with limited resources \cite{gholami2022survey}. In some situations, the compression rate attained by the quantization technique can be limited. For example, when 32-bit floating-point weights are transformed to 8-bit integer values, a maximum compression ratio of 4× is achieved, which may not be efficient for large models such as VGGNet \cite{VGG-16_original_paper}. Knowledge distillation aims to train a low-complexity student model using a complex teacher model. This approach does not exploit the redundancies within the model structure as efficiently as pruning \cite{model_compression_forDNN_survey}. 
  
  Pruning is another effective approach for model compression and optimization, as it removes unimportant or redundant parameters or computational units in a structured and an unstructured fashion \cite{liang2021pruning}. Unstructured pruning eliminates random insignificant weights, making it challenging to accelerate on general-purpose hardware \cite{he2023structured}. To address these challenges, specialized sparse matrix libraries and hardware accelerators have been developed. However, this represents an additional complexity to the model realization. Alternatively, structured pruning, such as filter or layer pruning, can entirely remove unimportant structures from the model.

  \begin{figure}
\centering
    \includegraphics[width=0.9\linewidth]{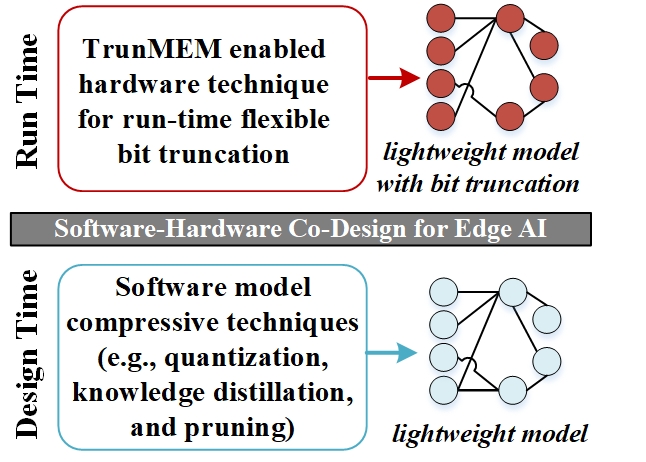}
  \caption{TrunMEM enabled software-hardware co-design for edge intelligence.}
  \label{fig:edge}
  \vspace{-0.2in}
\end{figure}

Unfortunately, all of those software techniques target model development and optimization during the design time. Once an optimized lightweight model is developed, the model cannot adapt to dynamic environments at runtime, limiting the model deployment process on edge devices. However, the adaptation of model performance is crucial for the deployment of edge models. For example, for a security camera, the classification performance requirement differs significantly depending on its deployment environment. In general, high classification performance is required for high-security environments with strict regulations such as airports, while medium accuracy may be acceptable for a lower-risk environment (e.g., an office) with further enhanced power efficiency. 

In order to adapt to the dynamic edge deployment environments, the bit truncation memory developed in this paper, i.e., TrunMEM, can be used simultaneously with the software techniques, as illustrated in Fig. \ref{fig:edge}. In our analysis, we considered three types of deep learning models for verifying bit truncation, including popular pre-trained CNN models (i.e., VGG-16, ResNet-56) for classification, an object detection model (i.e., Faster-RCNN), and custom pruned CNN models optimized for edge devices. For each CNN model, we used the channel attention-based filter pruning approach  in order to reduce the complexity of classification models. The chosen attention-based filter pruning approach was inspired by the Squeeze-and-Excitation channel attention method \cite{hu2018squeeze}, which showed elevated performance as compared to other filter pruning methods in \cite{liu2021channel,hu2022neural}. The developed lightweight models, together with their basic models, are deployed together with TrunMEM to enable performance-power adaptation in the edge deployment process, which will be discussed in detail in Section IV.

\section{Adaptive Bit Truncation}
\subsection{Bit Truncation and Optimal Value}

In order to enable low-power data storage, bit truncation needs to adapt the number of truncated bits to enable setting optimal truncation values. In terms of video data, which are represented by 8-bit integer pixel values, it has been concluded from previous work that setting the truncated LSBs to the mean value, i.e., 10...0 in binary will minimize the expected mean square error (MSE) \cite{edstrom2019content}. Here, to apply bit truncation to DNN weight storage, we study the truncated values for floating point numbers. For this study, IEEE 754 single precision floating point representation was used in our analysis, which has been used widely in deep learning systems \cite{edstrom2017data, IEEE_FloatingPointStandard}. To find an optimal value necessary to fill truncated values, we propose an approach that identifies the minimal error between all possible values. Consider the following:



Let $b_i\in \{0,1\}$ be the value of the $i^{th}$ bit, $i=0, 1, \cdots, 31$, where $b_{31}$ is the sign bit, $b_{23}, \cdots, b_{30}$ are the exponent bits, and $b_0, \cdots, b_{22}$ are the fraction bits. Thus, a number $y$ using the IEEE 754 standard can be represented as follows.
\begin{align} \label{IEEE754_Exp0}
    y & = (-1)^{b_{31}}2^{\sum_{i=0}^7 b_{23+i}2^i-127}\left(1+\sum_{i=0}^{22}b_{22-i}2^{-i-1}\right)
\end{align}
\begin{proposition}
    Let $T\subseteq \{0,1,\cdots, 22\}:=F$ be the index set of the truncation bits in which we only truncate the fraction bits. Assuming that the true value of $b_i, i=0,1,\cdots,31$ are evenly distributed, the best dummy value to set these truncated bits for minimizing the $E(MSE)$ is
    \begin{align}
        b_{t_{max}} = 1, b_j = 0, \forall j \in T - \{t_{max}\},
    \end{align}
    where $t_{max}$ is the maximum element in $T$.
\end{proposition}
\begin{proof}
    Let random variable $Y$ be the true decimal value of the number expressed by the IEEE 754 standard. By \eqref{IEEE754_Exp0}, we can rewrite $Y$ as
    \begin{align}
        Y &= (-1)^{b_{31}}2^{\sum_{i=0}^7 b_{23+i}2^i-127}\left(1+\sum_{k=0}^{22}b_k2^{k-23}\right) \quad (k := 22-i) \\
        & = (-1)^{b_{31}}2^{\sum_{i=0}^7 b_{23+i}2^i-127}\bigg(1+\sum_{k\in F-T} b_{k}2^{k-23} \notag\\
        & \quad \quad +\sum_{k\in T} B_{k}2^{k-23}\bigg) \notag \\
        &= c_1+c_1c_2+c_1\sum_{k\in T} B_{k}2^{k-23} \label{IEEE754_Trunc1}
    \end{align}
    where $B_{k}$ indicates the true binary value of the corresponding $b_{k} (i\in T)$, and $c_1:=(-1)^{b_{31}}2^{\sum_{i=0}^7 b_{23+i}2^i-127}$ and $c_2:=\sum_{k\in F-T} b_{k}2^{k-23}$ are both constant due to the definition of the truncation index set $T$. Let $m:=2^{|T|}$ be the number of possible combinations of the truncated bits and $x$ be the target (decimal) dummy value of $Y$. To minimize the $E(MSE)$, we are to minimize
    \begin{align*}
        f(x) = \frac{1}{m}\left[(x-y_1)^2+\cdots+(x-y_m)^2\right], 
    \end{align*}
    where $y_1, \cdots, y_m$ represents the (decimal) value of the number under each dummy value setting of the truncated bits. Thus, we get
    \begin{align*}
        0 = f'(x) = \frac{2}{m}\left(mx-\sum_{j=1}^m y_j\right)
    \end{align*}
    and hence
    \begin{align}
        x = \frac{1}{m}\sum_{j=1}^m y_j
    \end{align}
    Continuing the notation of \eqref{IEEE754_Trunc1}, denote
    \begin{align}
        y_j = c_1+c_1c_2+c_1\sum_{k\in T} b_{k,j}2^{k-23}
    \end{align}
    We have
    \begin{align}
        x &= \frac{1}{m}\sum_{j=1}^m y_j =  c_1 + c_1c_2 + \frac{c_1}{m}\sum_{j=1}^m\sum_{k\in T} b_{k,j}2^{k-23} \\
        &=  c_1 + c_1c_2 + c_1\sum_{k\in T} 2^{k-23} \left(\frac{1}{m}\sum_{j=1}^m b_{k,j}\right) \\
        &=  c_1 + c_1c_2 + \frac{c_1}{2}\sum_{k\in T} 2^{k-23} \label{IEEE754_Trunc_x2}
    \end{align}
    in which \eqref{IEEE754_Trunc_x2} holds because in $b_{k,j}, j=1, \cdots, m, \forall k$, there must be $\frac{m}{2}$ 1's and $\frac{m}{2}$ 0's due to the true-value evenly distribution assumption. Finally, comparing \eqref{IEEE754_Trunc_x2} with \eqref{IEEE754_Trunc1}, the best dummy value of the truncated fraction bits should be set to half of the value where $b_k=1, \forall k\in T$, and this is equivalent to setting the dummy value of the highest-level truncated bit to be one and all other truncated bits to be zero.
\end{proof}

For a simple example, if the indices of the truncated bits are $T=\{0,2,5\}$, then the best dummy value setting for these truncated bits in terms of minimizing the expected MSE is $b_5=1, b_2=b_0=0$.


\subsection{Proposed TrunMEM}

\begin{figure*}[ht]
	\begin{center}
    \includegraphics[width=\linewidth]{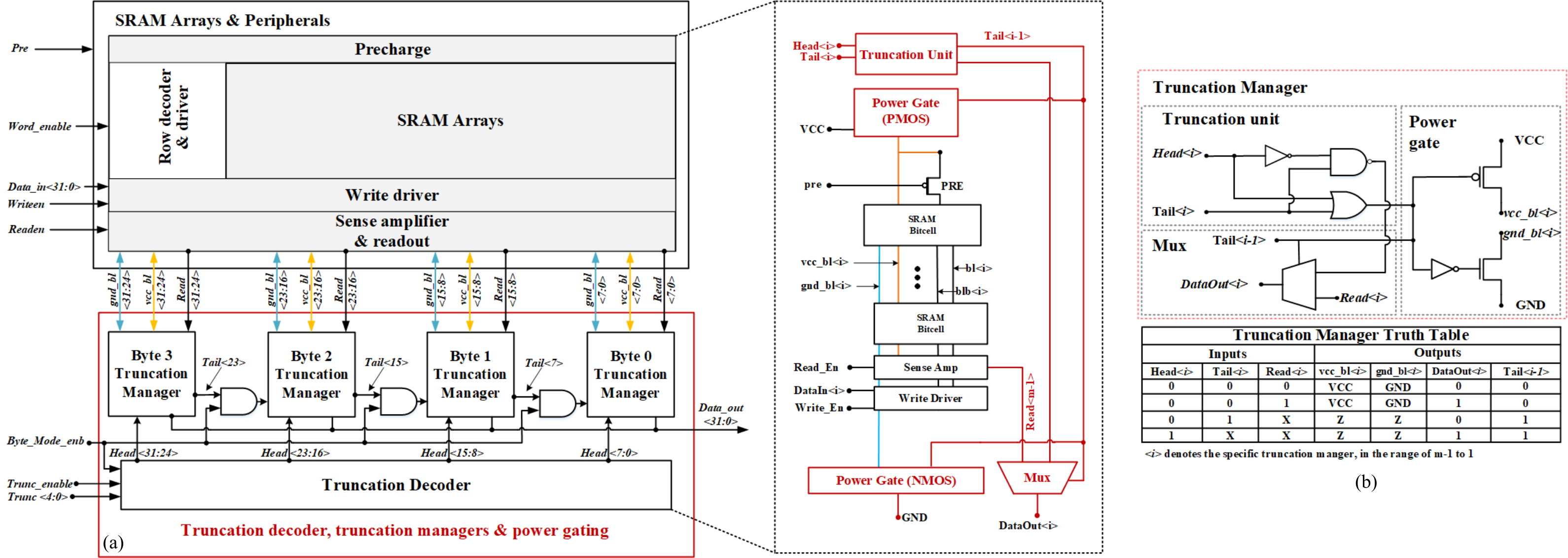}
	\end{center}
    \vspace{-0.1in}
	\caption{Proposed TrunMEM: (a) Memory structure and (b) Truncation manager circuitry and truth table. The components found within the Precharge, SRAM arrays, Row decoder \& driver, write driver, and sense amplifier are typical components found within common SRAM. The Truncation Managers are responsible for controlling the state of the data type (32bit float, or 8bit integer).}
    \label{proposed}
    \vspace{-0.1in}
 \end{figure*}



Fig. \ref{proposed} illustrates the architecture of the proposed TrunMem. As shown in Fig. \ref{proposed} (a), the SRAM array consists of $N$ words by $m=32$ bits, with each bitcell designed as a 6T SRAM. To achieve the primary objective of creating an adaptive and flexible truncation memory suitable for multiple applications, TrunMem introduces a truncation manager array. Truncation in each column of the memory array is facilitated by interconnected individual managers. 

Each proposed truncation manager employs CMOS power-gating circuitry to minimize the power consumption of the truncated bits. Power gating is applied to VCC and ground by connecting to virtual power rails or by leaving these virtual connections at a high impedance. Accordingly, each column of SRAM bitcells has its own dedicated pair of virtual rails, such that a single power gate can disconnect an entire column of SRAM bitcells at once, thus reducing power consumption. 

Peripheral circuitry operating on a specific column of bits will also be power-gated along with the SRAM bitcells for additional power savings. However, any SRAM bitcell disconnected from the power and ground rails quickly loses the value stored inside and cannot be read from. As such, the truncation manager circuit also populates read values when bit-lines get disconnected, which can either be a \textit{1} or \textit{0}.

Specifically, as shown in Fig. \ref{proposed} (b), each truncation manager consists of a truncation unit, power gating transistors, and an output multiplexer (Mux). The truncation manager has three operating states: normal operation without truncation, MSB of truncated bits (MSB truncated), and non-MSB of truncated bits (lesser bit truncated). For every column of bitcells, the signals to control the truncation state are the two inputs: $Head$ and $Tail$. $Head$ is used to detect whether it is MSB truncated and $Tail$ indicates whether it is lesser bit truncated. The operation process is as follows: If both the $Head$ or $Tail$ signal are disabled for a specific bit (e.g., \textit{ith} bit) is active, i.e., \textit{$Head<i> = Tail<i> = 0$}, the \textit{ith} bit will be in the normal operational state, and the $DataOut<i>$ will be the expected memory readout value from typical SRAM (i.e., $Read<i>$). During the normal state, the virtual rails ($vcc\_bl<i>$ and $gnd\_bl<i>$) remain connected to the supply voltage ($VCC$) and ground ($GND$), respectively. When $Head$ of the \textit{ith} bit is \textit{1}, i.e., \textit{$Head<i>=1$}, which indicates MSB truncated, $DataOut<i>$ will be \textit{1} and $vcc\_bl<i>$ and $gnd\_bl<i>$ will be placed into high impedance for power savings. Similarly, at the lesser bit truncated state, when $Tail<i> = 1$, virtual power rails enter high impedance and a $dataOut = 0$ value will be generated. 

As shown Fig. \ref{proposed}, in either truncation states, the output of the \textit{ith} bit, i.e., $Tail<i>$ will be applied to next bit as input, i.e., $Tail<i-1>$, and therefore the truncation managers of different bits work in series. As a result, $Head$ is the only external control signal required to be managed from external circuitry, while $Tail$ remains an internal signal. This series connection enables control such that if $Head<i> = 1$, the truncation manager signals to all lesser significant truncation managers to truncate, via the $Tail$ signals. Accordingly, the truncated memory returns an optimal truncation value with the most significant truncated bit as \textit{1} and other bits as zeros (i.e., 10...0\textsubscript{2}). The $Tail<m-1>$ is directly connected to $GND$, and the $Tail<0>$ signal is left floating in this design, as they serve no purpose.

A byte mode is introduced for flexibility in video applications. As shown in Fig. \ref{proposed}(a), dividing the truncation manager array into bytes and using an AND gate with the active LOW signal $Byte\_mode\_enb$ resets the $Tail$ value at each byte’s MSB, ensuring correct SRAM column truncation.

It can also be observed from Fig. \ref{proposed} (b) that the proposed truncation manager circuit is implemented with several simple logic gates, which does not induce a large area overhead. However, the sizing of power gating transistors can possibly be a point of concern, depending on the number SRAM word-lines they need to power. As a result, the number of truncation managers is a trade-off between flexibility and implementation cost. Full bit-adaptation is achieved in the design of Fig. \ref{proposed}, which adds $m$ truncation managers to the memory to support any value of truncation from 0 to $m-1$ bits. A flexible bit truncation provides the best flexibility to support different applications. However, to reduce the power-gating area overhead, if the target applications are known, the number of the truncation manager circuits may be reduced accordingly. Furthermore, a careful layout design of power gating transistors can reduce the implementation overhead, which will be discussed in Section IV.B.

\section{Experimental results}

\subsection{Experimental Methodology}\label{Experimental Methodology}

\textit{Hardware-Level Implementation and Verification:} To evaluate the effectiveness of the proposed memory, an SRAM with 1024 words$\times$32 bits was implemented using a \text{130nm} CMOS technology from SkyWater\cite{sky130}. EDA tools provided by SkyWater such as NGSPICE, Magic, and XSchem are utilized for physical design, implementations, and verifications. Based on the extracted parasitic parameters, we performed a comprehensive suite of post-layout simulations to evaluate functionality and performance parameters including timing diagram, power consumption, and silicon area overhead. 

\textit{Application-Level Evaluation: Videos.} We evaluated the video quality of TrunMEM using a custom Python-based emulator built in-house, following the same methodology as previous work \cite{ROI-AwareVideoStorage, edstrom2016luminance, edstrom2019content}. Specifically, 4,525 various videos from dataset \cite{HMDB_Dataset} in the YUV file format were selected as video inputs. Specifically, the emulator can read the pixel data of each video, perform analysis such as plain MB percentage or ROI, and then truncate the desired number of LSBs, according to a specific video technique \cite{ROI-AwareVideoStorage, edstrom2016luminance, edstrom2019content}. The generated video data is saved for display as video output. In our analysis, two widely-used metrics for video quality assessment, Peak Signal-to-Noise Ratio (PSNR) and Structural Similarity Index Measure (SSIM). These video quality measurements are paired with the power savings measurements produced by the hardware simulator to measure the trade-offs between quality and power savings. Our code of the video emulator is available at \cite{Git_Video_Emulator}.   

\textit{Application-Level Evaluation: DNNs.} 
To validate the effectiveness of TrunMEM, two widely used deep learning applications (i.e., classification and object detection) have been considered. To evaluate the performance of the models with TrunMEM, we also developed an in-house Python-based hardware emulator to truncate the weight and bias values for DNN inference. Our code of the DNNs emulator is available at \cite{Git_DNN_Emulator}. 

In addition to baseline models, as discussed in Section II, we also used the channel attention-based filter pruning approach to show the TrunMEM-enabled software-hardware co-design methodology. In the following subsections, we will discuss the methods used for each of the classification and object detection tasks.

\textbf{(1) Classification Task.} To evaluate the integration of TrunMEM and channel attention-based filter pruning in the context of image classification tasks, we selected two widely used deep neural network models (i.e., VGG-16 and ResNet-56). The models were chosen for their proven performance and widespread adoption in the computer vision community\cite{ zhang2015accelerating,ResNet50_Paper}. Two datasets, including CIFAR-10 and CIFAR-100, were used \cite{krizhevsky2009learning}. The CIFAR-10 dataset is considered a benchmark dataset in several deep learning applications. It consists of 60,000 32×32 true-color images categorized into 10 classes, ranging from airplanes and automobiles to birds, cats, and dogs [15]. In addition, the CIFAR-100 dataset consists of 100 classes, each with 600 images \cite{krizhevsky2009learning}. For both datasets, we have used 40,000, 10,000 and 10,000 images for training, validating and testing models, respectively. 

We then assessed the performance of the two models (i.e., VGG-16 and ResNet-56) on the two datasets (i.e., CIFAR-10 and CIFAR-100) prior to filter pruning and bit truncation and reported the performance of the models (i.e., considered as the baseline or benchmark performance). For the VGG-16 model trained on both CIFAR-10 and CIFAR-100 datasets, as well as ResNet-56 applied on CIFAR-10, we opted to use the stochastic gradient descent (SGD) optimizer with the categorical cross entropy loss function. However, when classifying CIFAR-100 using ResNet-56, we observed that the SGD does not support an optimal convergence of the model. Hence, we used the Adam optimizer along with the categorical cross entropy loss function to enhance the model performance.

The hyper parameters of the baseline models were chosen using the grid search algorithm \cite{dufour2019finite}. Using this algorithm, we exploited a wide range of learning rates from 0.000001 to 0.1, a large range of batch sizes from 16 to 256, and a limited range of training epochs from 50 to 250. As a result, we identified a batch size of 32 images for VGG-16 and ResNet-56 model training. Further, for training the VGG-16 model on both CIFAR-10 and CIFAR-100, we used an initial learning rate of 0.001 and a maximum of 150 epochs for model training. To enhance model convergence, we implemented a learning rate reduction approach reducing the initial learning rate to one-half of its value at the 50th and 100th training epochs. For the ResNet-56 model, we have trained the model on the CIFAR-10 dataset for 165 epochs with an initial learning rate of 0.1 and used the Piecewise Constant Decay scheduler. Subsequently, we followed a similar strategy for training the ResNet-56 on the CIFAR-100 dataset, utilizing a maximum of 200 epochs for model training and an initial learning rate of 0.001. To further improve model convergence, we employed a learning rate reduction scheme where the learning rate was decreased by a factor of 10 at the 80th, 120th, and 160th epochs. 

In addition to training the aforementioned baseline models above, we have applied channel attention-based filter pruning approach on the baseline models to generate lightweight models with a reduced number of learning parameters. The attention module was then attached to each convolutional layer in the baseline models and further retrained using the same previously mentioned training settings but for a different number of training epochs. For instance, we used 100 training epochs for retraining lightweight ResNet-56 on CIFAR-100 to avoid overfitting. We then evaluated the significance of the filters/feature maps in each layer of the models by measuring the scaling vector generated at the output of the channel attention module. We have then attempted to remove different percentages of the least significant filters within each layer of the models provided that the difference between the accuracy of the baseline and pruned models is less than or equal 4\%. For the evaluation of both the baseline/benchmark models and lightweight models, we have measured the model accuracy, number of learning parameters, model size, and number of floating point operations (FLOPs) required for the application of the model on a single image. Finally, the proposed bit truncation method was applied to the lightweight models and their performance was evaluated.

\textbf{(2) Object Detection Task.} In this study, we utilized the widely-used Faster-RCNN \cite{ren2015faster} model as a benchmark model and assessed its performance on the state-of-the-art PASCAL VOC-2007 dataset \cite{everingham2010pascal}. The dataset consists of two subsets (i.e., a training set of a total of 5,031 images and a testing set of 4,952 images). The dataset represents 20 distinct classes of objects contained in each image. In this study, a pre-trained VGG-16 has been used as the backbone of the Faster-RCNN model. The baseline model was then trained for 100 epochs using a learning rate of 0.001 and the stochastic gradient descent as the optimizer. 

For the object detection task, we used the intersection over union (IoU), the average precision (AP) and mean AP (mAP) to assess the performance of the Faster-RCNN model. Finally, the proposed TrunMEM was applied on the baseline Faster-RCNN model and the performance (i.e., AP, and mAP) of the model was measured.

\begin{figure*}[!htbp]
	\begin{center}
    \includegraphics[width=0.9\linewidth]{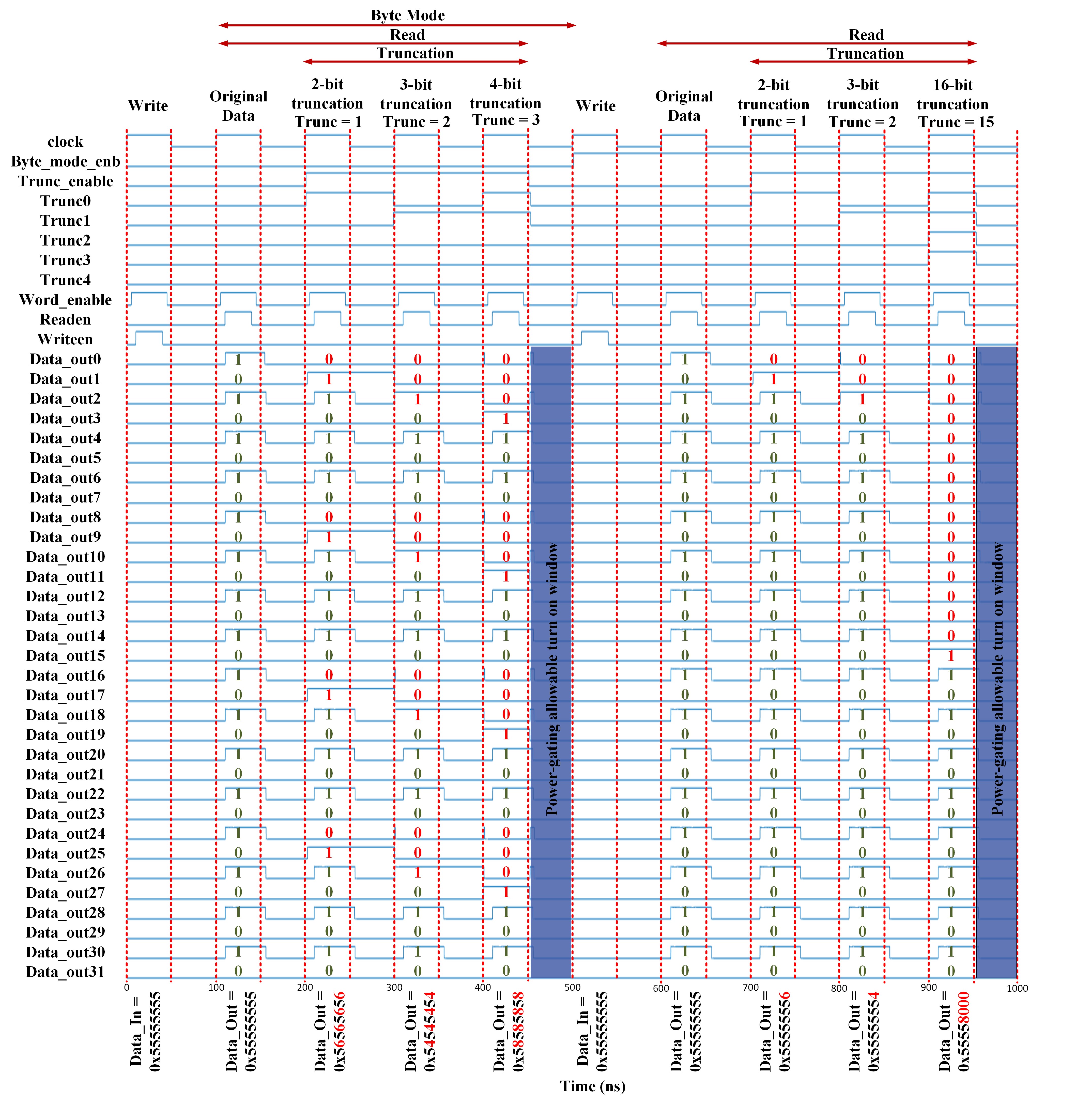}
  	\end{center}
	\vspace{-0.2in}
	\caption{Timing diagram of  TrunMEM. At byte mode, the binary value ``010101...01" is written to a random word, followed by four reading operations (starting with normal reading, 2-bit truncation, 3-bit truncation, and 4-bit truncation). Similarly, with the byte mode disabled, i.e., at the word mode, the value is written, followed by reads at 0-bit, 2-bit, 3-bit and 16-bit truncation.}\label{fig:sram_truncation_diagram}
 \vspace{-0.2in}
 \end{figure*} 

\subsection{Hardware-Level Post-Layout Simulation Results}\label{AA}
\textit{Timing Diagram:} Fig. \ref{fig:sram_truncation_diagram} shows the timing diagram of TrunMEM under different bit truncation conditions based on the post-layout simulations with parasitic extraction. The $clock$ signal doubles as the memory clock and the pre-charge (active-low) of the bit-lines for reading. $Word\_enable$, $Readen$, and $Writeen$ are used to activate the word-line for reading and writing access, and enable reading and writing, respectively. Truncation is enabled by $Trunc\_enable$ being HIGH, byte mode truncation is enabled by $Byte\_mode\_enb$ being LOW, and $Trunc<4:0>$ controls which bit-lines are truncated (i.e., if 1-bit truncation is desired, $Trunc$ is set to the binary value ``00000", else if 2-bit truncation is desired, $Trunc=``00001"$, and the series continues until $Trunc=``11111"$ meaning 32-bit truncation). The $Data\_out<31:0>$ bus is the output of the SRAM, with each output bit coming from a truncation manager as specified in Section \ref{proposed}. Functionality of the TrunMEM is tested by writing and reading to a random word with a clock period of 100ns. Specifically, in the first half of the simulation, the 32-bit binary value of "010101...01" is written into the selected word in the first clock cycle, then this word is read from in the following four clock cycles with different levels of bit truncation (i.e., 0, 2, 3 and 4 bit), with byte mode enabled. The same process is chosen for the second half of the simulation with byte mode disabled (i.e., the word mode enabled) and 0-bit, 2-bit, 3-bit and 16-bit truncation. The optimal truncated values generated by the truncation managers are highlighted in red for each $Data\_Out$. 

The power-gating switching delay, measured at 44 ns, is a factor worth noting but does not impact TrunMEM’s functionality due to its inherent design. When a bit-line powers down, its stored value becomes unreliable and is instead resolved by the truncation manager. Powering the bit-line back on occurs during a memory rewrite, with the condition that data must be rewritten when transitioning from higher to lower truncation levels. Thus, the transistor switching speed is inconsequential as long as it does not interfere with the write process, as highlighted in Fig. \ref{fig:sram_truncation_diagram}.

\textit{Power Efficiency:}
Power consumption is also measured from post-layout simulations, including parasitic extraction, at all possible bit truncation levels. For each truncation set-up, the average power consumption is measured to a random word with a clock period of 100ns. First, the selected word is initialized to the 32-bit hexadecimal value of $``A5A5A5A5"$, then $``FF00FF00"$ is written to the same word, and finally a read operation is executed on the same word. The initial condition and write value are selected to equally include all read/write memory operations. The average power consumption of the write operation during the first 100ns is constant across all simulations, and it is measured to be 2.35mW. The results for the power consumption during the read operation are shown in Fig. \ref{fig:Power} and Fig. \ref{fig:PowerByte} for word truncation mode, and byte truncation mode respectively. As expected, power savings increase as the number of truncated bit-lines increases. In the byte mode truncation test case, the power savings show a strong linear behavior. On average, the power consumption is reduced by 11.90\% or 296$\mu W$ with each additional bit truncated per byte. In the word truncation test case, the power savings increase linearly with a slight slope change per byte intervals, achieving an average of 2.87\% power savings per bit truncated, or 71$\mu W$. Byte truncation managers present a propagation delay of 0.5-4ns when truncating a bit-line and setting the DataOut value based on Section \ref{proposed}. Creating an input dependence for power consumption. In the post layout simulations, alternating bytes have either inputs of $``00"$ or $``FF"$ to measure the impact on power consumption of this input dependency. The results show that an average of 90 $\mu W$ per truncated bit is saved (or 3.6\% power savings) on the bytes where the input is $``00"$, while 60 $\mu W$ (or 2.2\% power savings) are saved when the input is $``FF"$. The results are reflected on the change is slope of Fig. \ref{fig:Power} and prove that, even with this data-input dependency, power consumption decreases with each additional bit truncated. Also, the power overhead caused by the additional circuitry of TrunMem, such as Truncation controller and power-gate, is evaluated, which increase the total power consumption by ~1.1 $\mu W$ (or 0.47\% of total power consumed) as compared to the basic SRAM design.  

\begin{figure}
    \centering
    \subfloat[Continuous truncation mode]{\includegraphics[width=\linewidth]{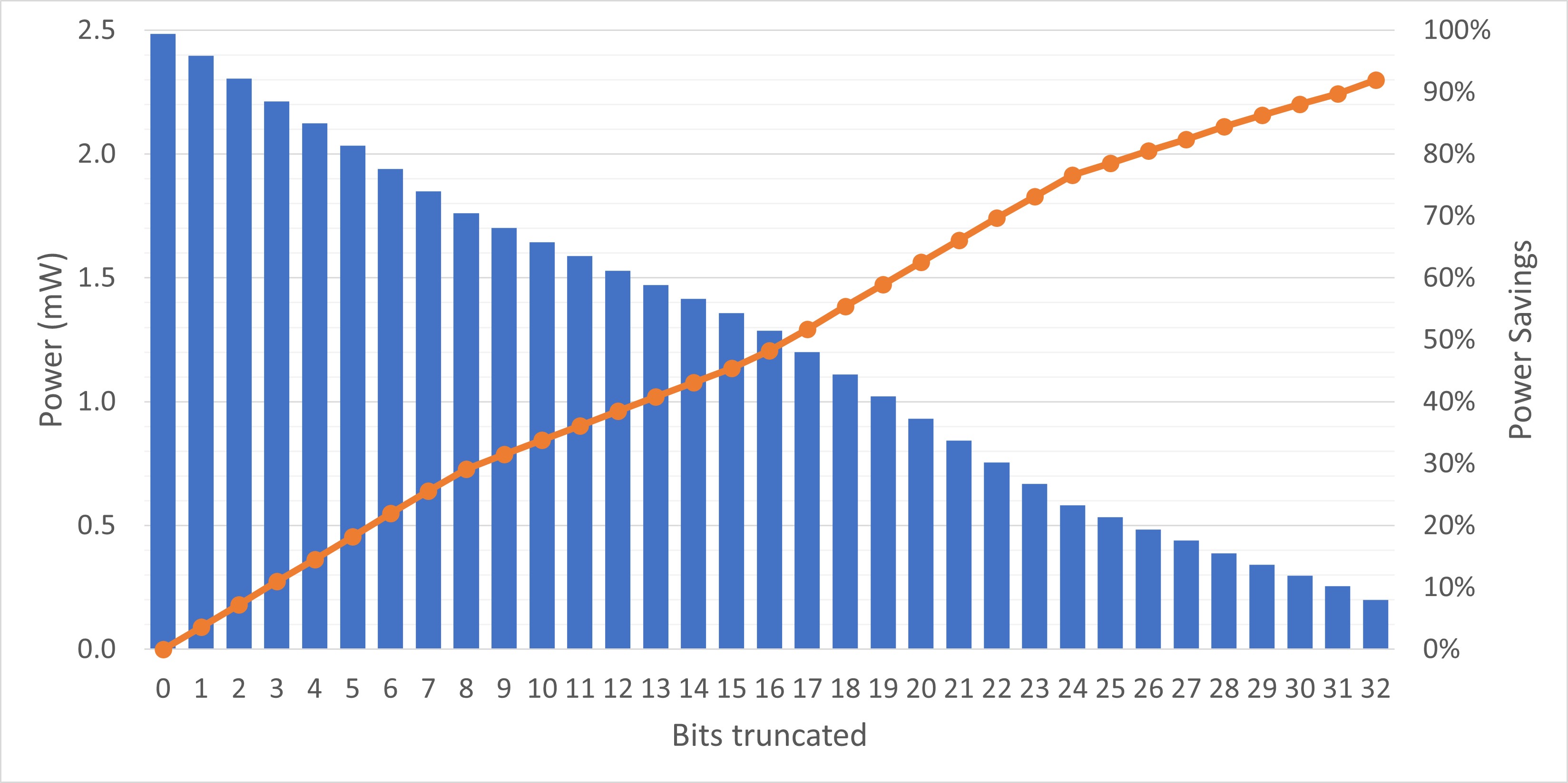}\label{fig:Power}}
    \\
    \subfloat[Byte mode]{\includegraphics[width=\linewidth]{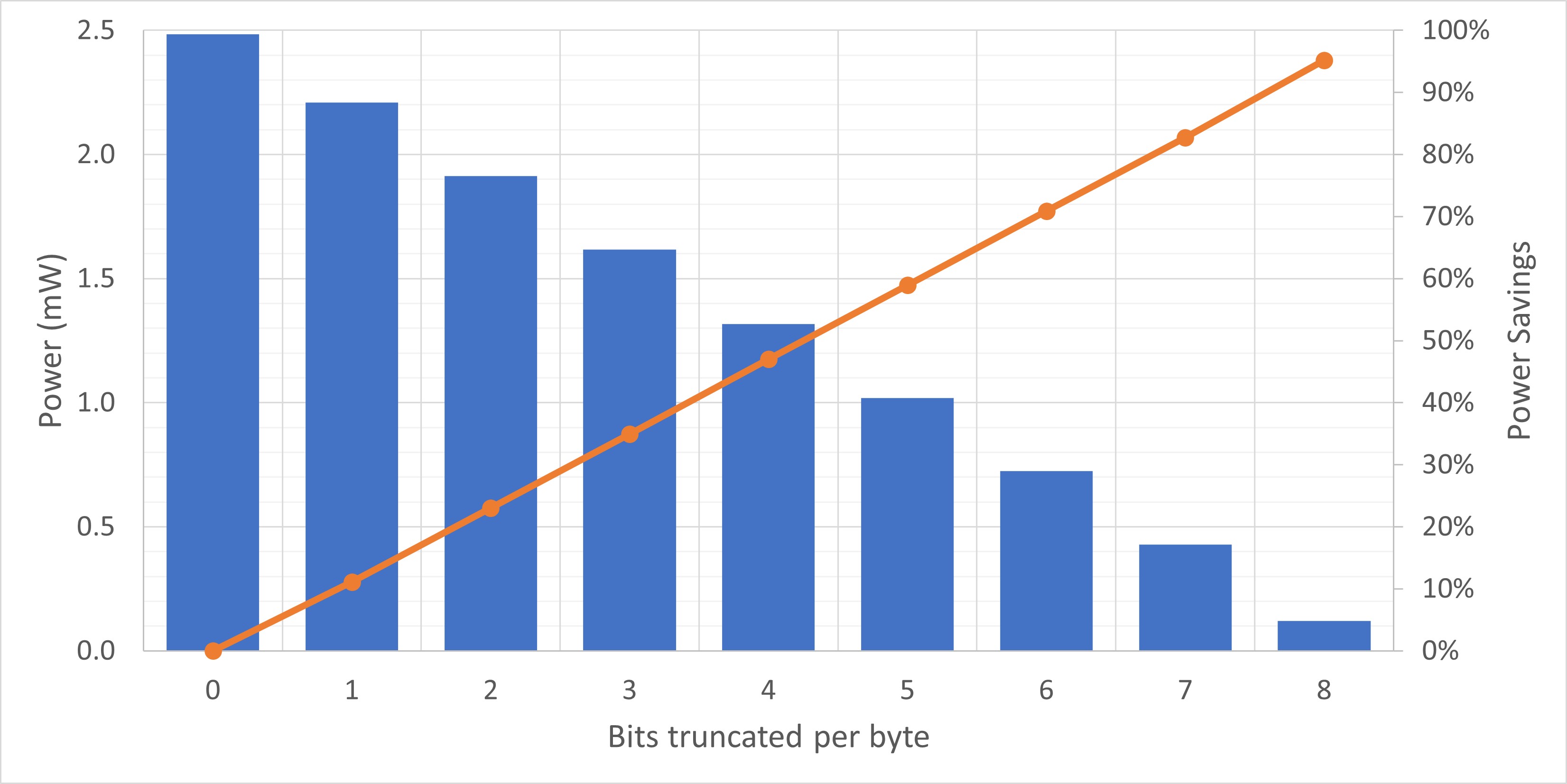}\label{fig:PowerByte}}
    \caption{Power consumption of TrunMEM at each level of bit truncation on a 32x1024 memory array. Using Xscheme \cite{Xschem}, NGspice \cite{NGspice}, and the Sky130 PDK \cite{sky130}}
    \vspace{-0.2in}
\end{figure}

\textit{Implementation cost}:
The layout of the TrunMEM is shown in Fig. \ref{fig:layoutdesign}, where the TrunMEM is highlighted from a multi-project design. The control circuitry of the TrunMEM (i.e. input register, word-line decoder, truncation decoder, and output multiplexer) is integrated using an automated RTL-to-GDSII flow, along with control circuitry of other projects. The exact position and area of the RTL control circuitry is not know, therefore it is not used for overhead calculations, but a visual estimation is highlighted on Fig. \ref{fig:layoutdesign}. The silicon area overhead of the TrunMEM is caused by the added 32 bit truncation managers. These truncation managers are divided into 4 byte managers, while area and placement is optimized to fit the width of the 1024 words $by$ 32 bits 6T SRAM. The main source of area consumption inside a single truncation manager are the power gate transistors; to minimize area, they are implemented using a particular finger-based design approach. As observed in Fig. \ref{fig:layoutdesign}, each power transistor has a width of 16 µm and it is implemented with eight fingers. The silicon area consumption of the truncation managers is 15,210 $\mu m^2$, which represents only 2.89\% compared to the 525,915 $\mu m^2$ total area of TrunMEM. Since each bit-line is managed by a single truncation manager, the area overhead of TrunMEM decreases as the number of memory words increases, with the only constraint being the switching speed of the power-gating transistors. 

\begin{figure*}[ht]
\centering
\includegraphics [width=\linewidth]{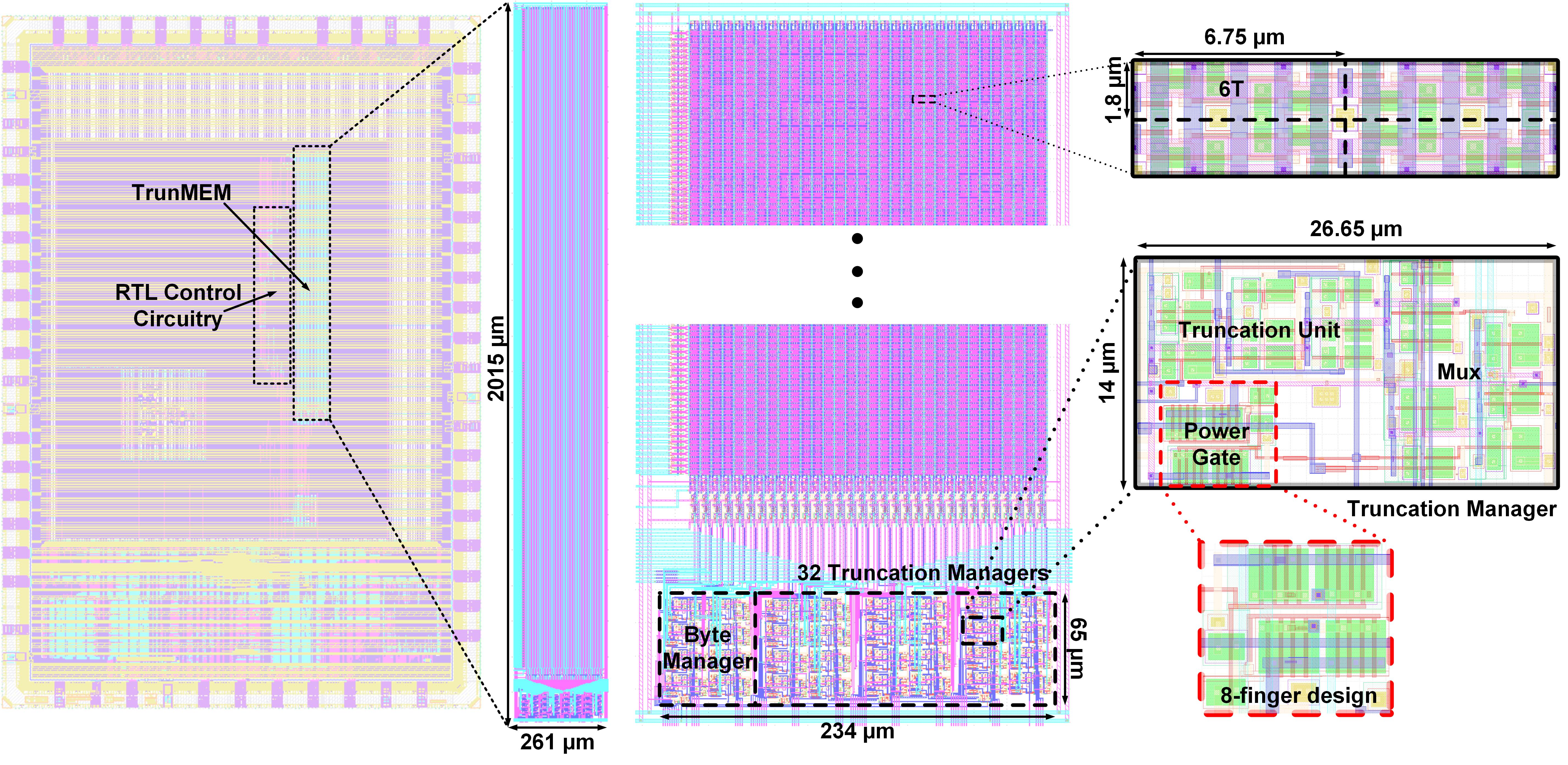}
\vspace{-0.3in}
\caption{TrunMEM layout, sent to fabrication as part of a multi-project design on Skywater 130nm technology \cite{sky130}. The TrunMEM is designed as a 32x1024 6T SRAM with 32 bit truncation managers divided into 4 byte managers.}
\label{fig:layoutdesign}
\vspace{-0.2in}
\end{figure*}

\subsection{Application-Level Evaluation Results - Videos}\label{Application_Evaluation1}


Table \ref{table:VideoAverages} lists the statistical results (Averages ± Standard Deviation) of output quality and power savings of 4,525 videos with TrunMEM. It can be seen that the Content-aware \cite{edstrom2019content} and ROI-aware \cite{ROI-AwareVideoStorage} video systems produce fairly high video quality measured by the peak signal-to-noise ratio (PSNR) and the structural similarity index (SSIM). The ROI-aware method, on average, achieves the highest video quality by PSNR (51.60±2.11dB), and SSIM (99.66±0.29\%) for any given video (Table \ref{table:VideoAverages}). In terms of power savings, among three video techniques, the luminous-Aware method \cite{edstrom2016luminance} with TrunMEM achieves the highest power savings (34.93-47.02\%) while sacrificing video quality. 14±5.28\% and 20.95±6.08\% power savings are enabled by the content-aware \cite{edstrom2019content} and ROI-aware \cite{ROI-AwareVideoStorage}, respectively.

Another important observation from Table \ref{table:VideoAverages} is that the Luminous-Aware method does not present any variance in the enabled power savings. Because, for any given video, the ambient illumination is assumed not to change, thus locking the entire video to either 3 or 4 LSBs truncated across the entire video.



To further evaluate the different video techniques with TrunMEM, the visual output quality of several representative videos is illustrated in Fig. \ref{Fig:VideoVisual}. Specifically, the videos 'Faith Rewarded Talk', 'Ian Poulter Golf', and 'Famille Des Poussins' are presented as these videos show the range of possible truncation levels (low, medium, and high variance) categorize as defined by \cite{edstrom2019content}. The visual differences between the MB and ROI methods are hard to distinguish in every case, except 'Famille Des Poussins'. From this frame comparison, it can be seen that the woman's facial features have noticeable color distortion in the MB method. For the Luminous Aware methods (overcast and sunlight conditions), the color distortion becomes very apparent \cite{edstrom2016luminance}. Visible facial features have notifiable distortion in overcast, and heavily distorted in sunlight. Looking closely at 'Ian Poulter Golf', it can be seen how the overcast and sunlight conditions comparisons distort the slow color changes across the grass, and green wall in the background. For the sunlight image this color distortion is amplified and is noticeable. It should be noted however, as PSNR and SSIM video quality metrics show the most distortion for the Luminous aware method, it also saves the most power consumption compared to the other two methods (Fig. \ref{Fig:VideoVisual}).

From the above analysis, it can be concluded that, the state-of-the art video techniques, including Content-aware \cite{edstrom2019content} and ROI-aware \cite{ROI-AwareVideoStorage}, have its benefits for a specific video application. The proposed TrunMEM can support all of those different video applications with enhanced power efficiency, although the enabled power savings of each technique is technology-dependent. Accordingly, TrunMEM makes it feasible to adapt from three different video techniques in real-time, in order to optimize the power consumption or video quality according to the requirement of the applications.

\begin{figure*}[t]
\centering
\includegraphics[width=0.99\linewidth]{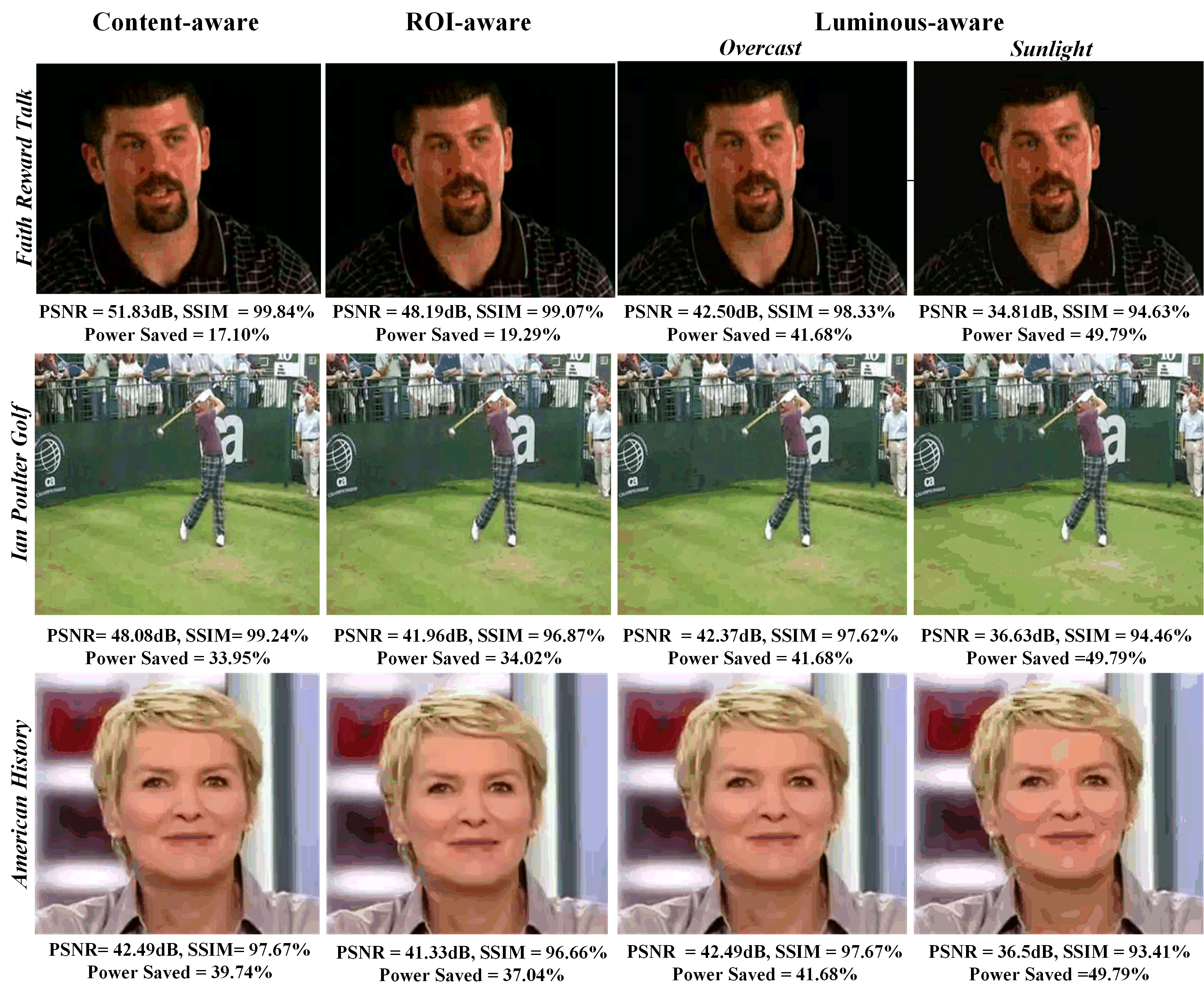}
\caption{Visual output quality with TrunMEM. Source videos 
from \cite{HMDB_Dataset, YoutubeM8}. 'Faith Rewarded Talk' is classified as a low variance, 'Ian Poulter Golf' is medium variance, and 'Famille Des Poussins' is high variance using the content-aware technique \cite{edstrom2019content}. }
\label{Fig:VideoVisual}
\end{figure*}



\begin{table}[ht]
\centering
\caption{Video Quality and Power savings of state-of-the art video systems with TrunMEM (averages ± standard deviation). Sample size = 4,525 videos from dataset \cite{HMDB_Dataset}}
\label{table:VideoAverages}
 \resizebox{\columnwidth}{!}{%
\begin{tabular}{l|c|c|c|}
\cline{2-4}
                                    & \multicolumn{1}{l|}  {PSNR(dB)} & \multicolumn{1}{l|}{SSIM(\%)} & \multicolumn{1}{l|}{Power Saved(\%)} \\ \hline
\multicolumn{1}{|l|}{content-aware}     & 51.60±2.11                    & 99.66±0.29                    & 14±5.28                           \\ \hline
\multicolumn{1}{|l|}{ROI-aware}    & 43.15±1.61                    & 97.30±0.75                    & 20.95±6.08                          \\ \hline
\multicolumn{1}{|l|}{luminance-aware(overcast)} & 42.32±0.57                    & 97.45±0.35                    & 34.93±0                                \\ \hline
\multicolumn{1}{|l|}{luminance-aware(Sunlight)} & 36.38±0.68                    & 93.87±1.06                    & 47.02±0                                \\ \hline
\end{tabular}
}
\end{table}



%

\subsection{Application-Level Evaluation Results - DNNs}\label{Application_Evaluation2}

\textit{DNN Categorization Task:} 
We evaluated the DNN performance using a custom in-house hardware emulator, emulating the effects of TrunMEM, utilizing TensorFlow as the base DNN processing library \cite{tensorflow2015-whitepaper}. We have reported the performance, number of Floating Point Operations (FLOPs), and number of learning parameters for the baseline and lightweight (i.e., when attention-based filter pruning was solely applied on the models) VGG-16 and ResNet-56 models applied on the CIFAR-10 and CIFAR-100 datasets in Table \ref{Table:model size}. 

It is clear from Table \ref{Table:model size} that significantly lower-complexity VGG-16 and ResNet-56 models were realized using the channel attention based filter pruning technique at a slight reduction in the model accuracy as compared to the baseline models. For example, for the VGG-16 model applied on the CIFAR-10 dataset, we were able to achieve approximately 96\% and 90\% reduction in the number of learning parameters and FLOPs respectively with 2.73\% decrease in the model accuracy using the channel attention based filter pruning technique \cite{liu2021channel}. For the CIFAR-100 dataset with higher number of classes and lower number of samples per each class, we were able to achieve at least 87\% and 74\% reduction in the number of model parameters and FLOPs, respectively, for the filter pruned model, at the cost of a tolerable loss of model accuracy (i.e., almost 4\%). For the ResNet-56 architecture, we achieved a reduction of nearly 75\% in the number of parameters and FLOPs when applied on the CIFAR-10 dataset with solely 2\% accuracy loss. Similarly, for the CIFAR-100 dataset, a reduction of approximately 62\% in the number of learning parameters and 52\% in the number of FLOPs were achieved with almost a 3\% accuracy drop. 


\begin{table}[htbp]
\centering
\caption{Size of allocated memory, FLOPS, and accuracy for baseline (BL), and lightweight (LW) models.}
 \resizebox{\columnwidth}{!}{%
\begin{tabular}{|l|ll|ll|}
\hline
                               & \multicolumn{2}{l|}{VGG-16}               & \multicolumn{2}{l|}{ResNet-56}            \\ \hline
Datasets                       & \multicolumn{1}{l|}{CIFAR-10} & CIFAR-100 & \multicolumn{1}{l|}{CIFAR-10} & CIFAR-100 \\ \hline
Baseline Model Accuracy        & \multicolumn{1}{l|}{93.73}    & 72.17     & \multicolumn{1}{l|}{91.35}    & 69.95     \\ \hline
Baseline Model   FLOPs         & \multicolumn{1}{l|}{705M}     & 705.8M    & \multicolumn{1}{l|}{126.2M}   & 126.2M    \\ \hline
Baseline Model   Parameters    & \multicolumn{1}{l|}{15.24M}   & 15.29M    & \multicolumn{1}{l|}{0.86M}    & 0.867M    \\ \hline
Lightweight   Model Accuracy   & \multicolumn{1}{l|}{91.0}     & 68.13     & \multicolumn{1}{l|}{89.63}    & 67.02     \\ \hline
Lightweight   Model FLOPs      & \multicolumn{1}{l|}{68M}      & 180M      & \multicolumn{1}{l|}{31.82M}   & 61.1M     \\ \hline
Lightweight   Model Parameters & \multicolumn{1}{l|}{0.68M}    & 2M        & \multicolumn{1}{l|}{0.21M}    & 0.33M     \\ \hline
FLOPs Reduction                & \multicolumn{1}{l|}{90.29\%}  & 74.43\%   & \multicolumn{1}{l|}{74.80\%}  & 51.5\%    \\ \hline
Parameters   Reduction         & \multicolumn{1}{l|}{95.52\%}  & 86.52\%   & \multicolumn{1}{l|}{74.65\%}  & 61.7\%    \\ \hline
\end{tabular}
\label{Table:model size}
}
\end{table}

Next, TrunMEM was applied on both the baseline and lightweight models (i.e., post filter pruning). Fig. \ref{fig:classification_performance} shows the performance of both the baseline VGG-16 and ResNet-56 models and their corresponding lightweight versions, trained on CIFAR-10, and CIFAR-100 across all possible levels of truncation.

\begin{figure*}[t]
\centering
\includegraphics[width=0.65\linewidth]{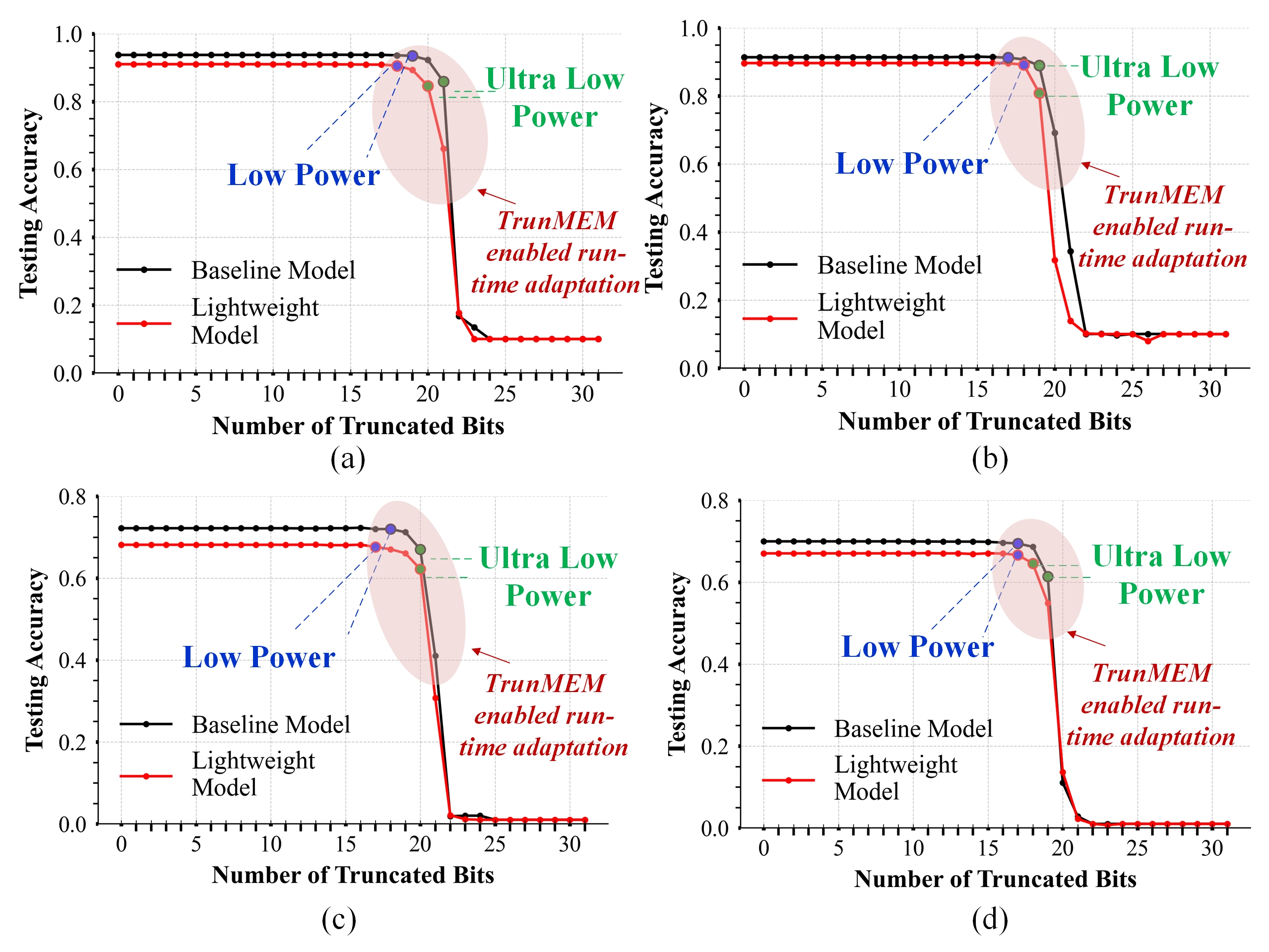}
\caption{Performance of baseline and Lightweight models with software (pruning) and hardware (TrunMEM) techniques: (a) VGG-16 models with CIFAR-10 dataset; (b) ResNet-56 models with CIFAR-10 dataset; (c) VGG-16 models with CIFAR-100 dataset; and and (D) ResNet-56 models with CIFAR-100 dataset. On each graph, "Low power" mode ensures validation accuracy remains within 0.5\%, and "Ulta Low Power" ensures validation accuracy remains within 10\% of the original accuracy}\label{fig:classification_performance}
 \end{figure*}

From Table \ref{Table:model size}, it can be seen that both the lightweight models with and without truncation provide a significant advantage over the baseline models in terms of the size of allocated memory where up to $22\times$ reduction in memory for the VGG-16 model that was trained and tested on CIFAR-10 dataset was achieved. While Fig \ref{fig:classification_performance} shows nearly identical model performance drop-offs between the baseline and lightweight models. Showing how design time memory optimizations work independently of the TrunMem run time power optimizations.

As shown in Fig. \ref{fig:classification_performance}, it is clear that both the performance of baseline VGG-16 and ResNet-56 models and associated filter-pruned lightweight models exhibited resilience to bit truncation, with optimal settings for 'Low Power' (blue) and 'Ultra Low Power' (green) varying between 16-20 bits truncated; achieving 44.17\%-61.42\% power savings. At the 'Low Power' setting, the model performs within 0.5\% accuracy, at this point, no noticeable accuracy drop could be observed at run-time. 'Ulta Low Power' however, does drop by as much as 10\%, and should only be used if power consumption becomes more important than model performance. Thus, based on the lightweight model enabled by software techniques (Table \ref{Table:model size}), TrunMem is effective in enabling run-time adaptation with power savings. 


\textit{Object Detection Task:} 
Fig. \ref{fig:FasterRCNN_mAP} illustrates the baseline average precision (AP) across 20 classes and the overall mean average precision (mAP) of the Faster-RCNN model as we vary the number of truncated bits for the learning parameters of the model. The highest AP of 0.82 was achieved for the car class, while the lowest AP of 0.405 was observed for the potted plant class. In addition, a mAP of 0.663 was obtained across all classes. It is also clear that the model performance (i.e., mAP) remained almost unchanged up to 17 bits. Beyond this point, the performance started to decline, and by 22 bits, the model completely lost its ability to detect any objects in an image (i.e., mAP of 0).

\begin{figure}[htbp]
	\begin{center}
  \includegraphics[width=0.9\linewidth]{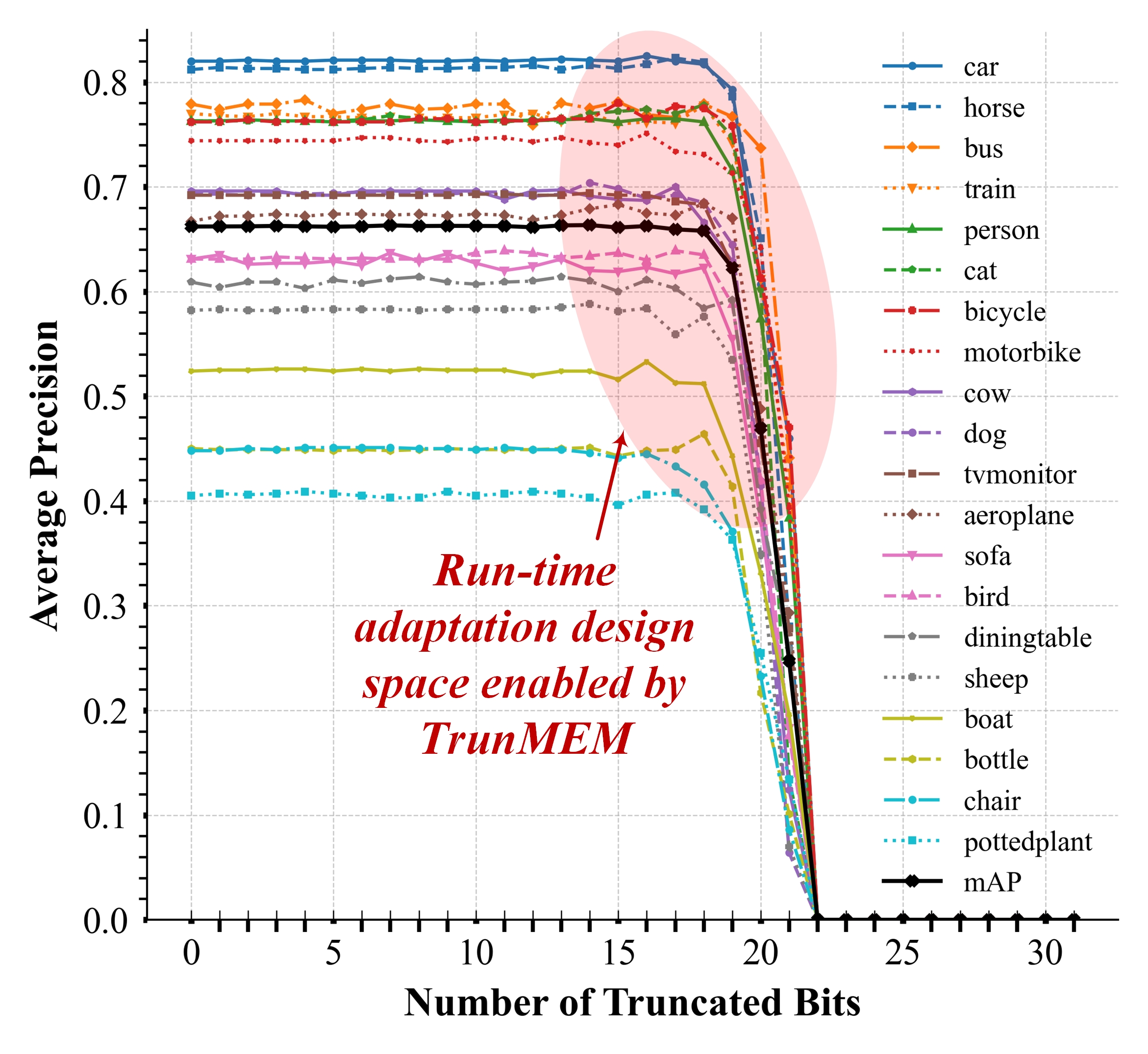}
	\end{center}
	\vspace{-0.2in}
	\caption{Performance of Faster-RCNN models with TrunMEM on PASCAL VOC-2007 dataset. }\label{fig:FasterRCNN_mAP}
    
 \end{figure}

 Further, Fig. \ref{fig:FastRCNN_Visual} shows a sample image of the PASCAL VOC-2007 test dataset, along with the bounding box predictions generated by the baseline Faster-RCNN model and several bit-truncated versions of the same model. Moreover, the prediction confidence for the predicted class is also reported in Fig. \ref{fig:FastRCNN_Visual}. For the horse class, the 17-bit truncated model produced a bounding box with an IoU of 0.992, with power savings of 51.69\% while the 21-bit truncated model achieved an IoU of 0.706, and power savings of 66.08\%. For the person class, the 17-bit truncated model resulted in a bounding box with an IoU of 0.991, and the 21-bit truncated model yielded an IoU of 0.852. Therefore, at a 21-bit truncation, the model can still predict the bounding box for the horse and person class with a slight reduction in the prediction area. However, at a 22-bit truncation, the model failed to detect the horse or the person where the IoU was found to be 0.

\begin{figure}[htbp]
	\begin{center}
  \includegraphics[width=\linewidth]{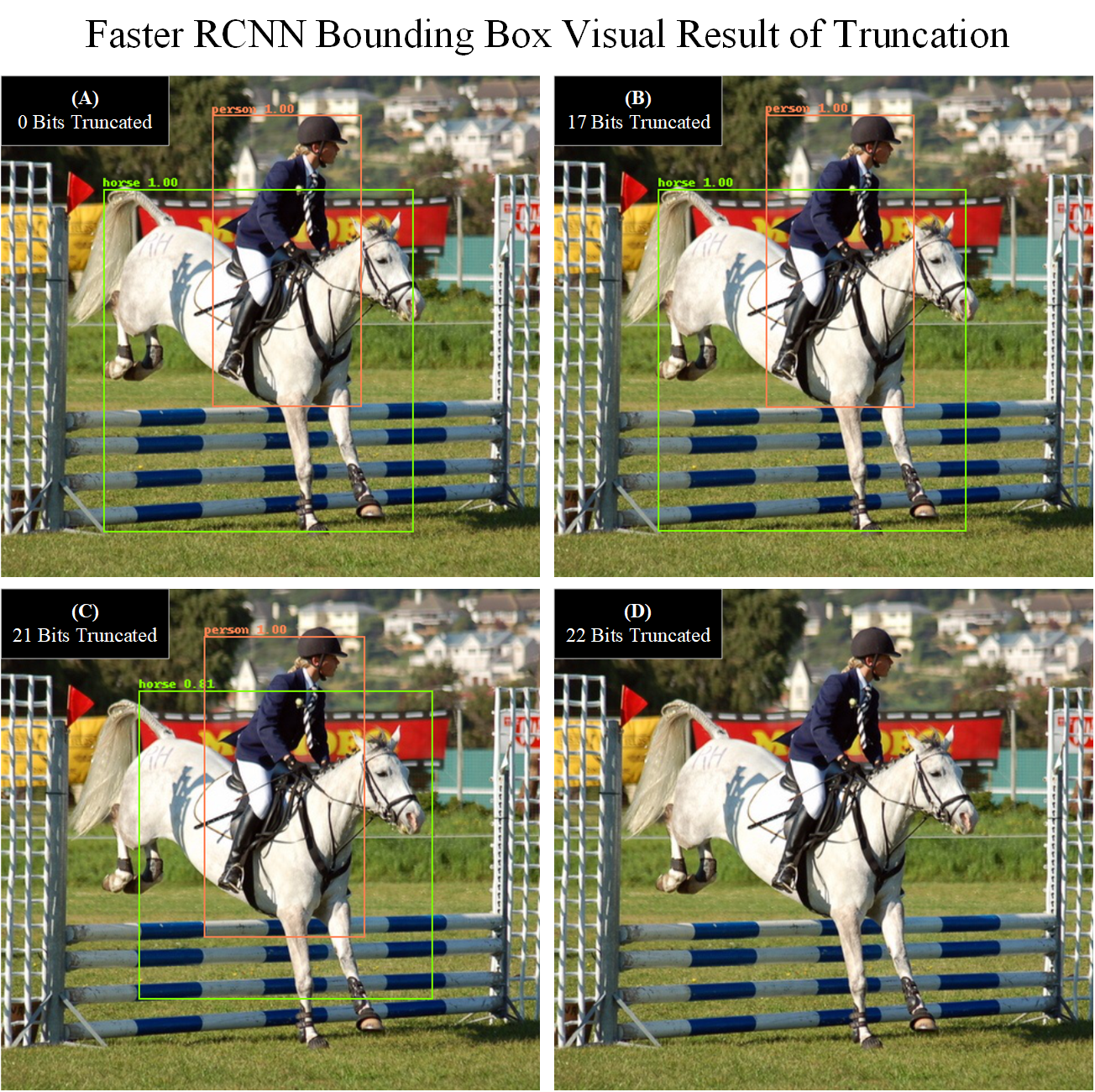}
	\end{center}
	\vspace{-0.2in}
	\caption{Fast RCNN bounding boxes at (A) No truncation. (B) 17 bits truncation. (C) 21 Bits truncated, and (D) 22 bits truncated. It can be observed that the results between 0 and 17 bits truncated are perfectly identical. Whereas, 21 bits truncated the boxes for both the person and horse are misaligned, and classification for the horse also decreases to 81\% confidence. At 22 bits truncated, total failure occurs, with no boxes produced.}\label{fig:FastRCNN_Visual}
 \vspace{-0.2in}
 \end{figure}

 \begin{table*}[!b]
\caption{Comparison with the state-of-the art}
\begin{tabular}{|l|c|c|c|c|}
\hline
                             & \begin{tabular}[c]{@{}c@{}}luminance-aware   \\ \cite{chen2015vcas, edstrom2016luminance,chen2018viewer}\end{tabular} & \begin{tabular}[c]{@{}c@{}}content-aware   \\ \cite{edstrom2019content}\end{tabular} & \begin{tabular}[c]{@{}c@{}}ROI-aware   \\ \cite {ROI-AwareVideoStorage} \end{tabular} & \textbf{This work}\\ \hline
                              
truncation mode & byte mode                       & byte mode              & byte mode                                    & byte mode and continuous truncation mode \\ \hline

\begin{tabular}[c]{@{}l@{}}number of truncated bits   \\ (byte mode)\end{tabular} & 3 or 4 LSBs                      & 0-4 LSBs              & 0 or 3 LSBs                                    & 0-8 LSBs (full flexibility)  \\ \hline


\begin{tabular}[c]{@{}l@{}}enabled maximum   \\ power savings with videos\end{tabular} & \begin{tabular}[c]{@{}c@{}} 32.6\% \\ (4 LSBs truncated) \end{tabular}  & \begin{tabular}[c]{@{}c@{}} 33.31\% \\ (4 LSBs truncated) \end{tabular} & \begin{tabular}[c]{@{}c@{}} 19.74\% \\ (3 LSBs truncated) \end{tabular}  & \begin{tabular}[c]{@{}c@{}}47.02\% \\ (4 LSBs truncated) \end{tabular} \\ \hline

\begin{tabular}[c]{@{}l@{}}number of truncated bits   \\ (continuous truncation mode)\end{tabular} & $\times$                      & $\times$               & $\times$                                    & 0-32 (full flexibility)  \\ \hline

\begin{tabular}[c]{@{}l@{}}enabled optimal   \\ power savings with DNN\end{tabular} & $\times$                       & $\times$              & $\times$                                   & 51.69\% @ 17bits  \\ \hline

\begin{tabular}[c]{@{}l@{}}silicon area overhead   \\ over basic design \end{tabular}      & $<$0.01\%                       & 0.32\%              & \begin{tabular}[c]{@{}c@{}}-   \\ (no layout implementation) \end{tabular}                                & 2.89\%  \\ \hline
\end{tabular}
\label{Table:comparison with other work}
\end{table*}

\begin{figure}
\centering
  \includegraphics[width=0.8\linewidth]{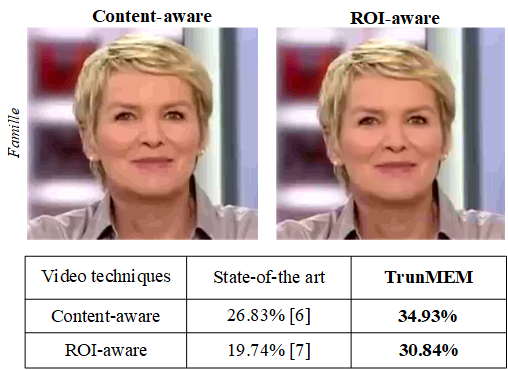}
  \caption{Power efficiency comparison with the state-of-the art. Famille video available \cite{YoutubeM8} (video tag: wF6lvdXXwc4)}
  \label{fig:power}
\end{figure}

\subsection{Comparisons to existing works}

Table \ref{Table:comparison with other work} compares TrunMEM with state-of-the art truncation memory designs. As shown, the proposed TrunMEM enables full flexibility in running-time quality adaptation. First, luminance-aware \cite{chen2015vcas, edstrom2016luminance,chen2018viewer}, content-aware \cite{edstrom2019content}, and ROI-aware \cite{ROI-AwareVideoStorage} can only enable byte-mode truncation and therefore cannot support the word-based data storage, such as deep learning systems. In addition, in terms of the byte mode truncation, existing designs can only enable limited number of truncated bits: in each byte, luminance-aware \cite{chen2015vcas, edstrom2016luminance,chen2018viewer} truncates 3 or 4 LSBs, content-aware \cite{edstrom2019content} truncates 0 to 4 LSBs, and ROI-aware \cite{ROI-AwareVideoStorage} truncates 0 or 3 LSBs. As a result, existing memories can also be used for a specific technique. Alternatively, TrunMEM is capable of truncating the entire range of bits, i.e., 0-8 per byte in the byte mode and 0-32 per word in the word mode (with the word size of 32). Although such flexibility comes at a cost of 2.89\% silicon area overhead, which is higher than existing designs, the overhead can be reduced with more words added to the memory, as discussed in Section IV-B. Accordingly with TrunMEM, those different viewer-aware video techniques may be applied simultaneously to achieve power optimization. 

In terms of power savings with the same number of truncated bits, TrunMEM achieves 47.02\% power savings with 4 truncated LSBs, which is much higher than existing designs in \cite{chen2015vcas, edstrom2016luminance,chen2018viewer,edstrom2019content}. ROI-aware \cite{ROI-AwareVideoStorage} achieves 19.74\% power savings as it only truncate 3 LSBs.  Since the content-aware and ROI-aware techniques also strongly depend on specific videos, we also compared the power effectiveness of those two techniques using the same video 'Famille'. The results are shown in  Fig. \ref{fig:power}. It can be seen that both techniques enable additional power savings with TrunMEM compared to existing designs.

From the above analysis, it can be concluded that TrunMEM enables run-time quality adaptation with full flexibility as well as significantly enhanced power efficiency compared to existing techniques.

\section{Conclusion}
This paper has presented a novel quality-adaptive bit-truncation memory design - TrunMEM to support different approximate applications, such as video and deep learning. TrunMEM enables up to 47.02\% and an optimal 51.69\% power savings in videos and DNN inference, respectively. With full truncation flexibility, TrunMEM can also be applied to general approximate applications. In particular, TrunMem demonstrates strong promise for general applications into application-specific hardware platforms such as GPUs and edge devices.

Since TrunMEM can support three different view-aware video techniques simultaneously, future work will develop a smart video system which can adapt the truncation method depending on both viewing surroundings (e.g., overcast, sunlight, dark) and video content characteristics (e.g., ROI), thereby further optimizing power efficiency. We also plan to conduct chip-based testing and perform silicon data analysis in our future work after the chips are received. Finally, we plan to further study the relationship
between the hardware-based run-time truncation in the DNN inference process and software-based compression (e.g., pruning, quantization) in the model training process to better adapt the trade-off between power efficiency and quality (e.g., accuracy) to meet the requirements of a variety of AI tasks on power-constraint edge devices.


\section*{Acknowledgment}
This work was supported in part by the National Science Foundation under OIA-2218046, CNS-2211215, ECCS-2420994, CCF-224734, EES-2427766, and OIA-2428981, the U.S. Department of Energy under DE-SC0025561, and Alabama EPSCoR Graduate Research Scholars Program.

\bibliography{ref}

\begin{thebibliography}{10}
\providecommand{\url}[1]{#1}
\csname url@samestyle\endcsname
\providecommand{\newblock}{\relax}
\providecommand{\bibinfo}[2]{#2}
\providecommand{\BIBentrySTDinterwordspacing}{\spaceskip=0pt\relax}
\providecommand{\BIBentryALTinterwordstretchfactor}{4}
\providecommand{\BIBentryALTinterwordspacing}{\spaceskip=\fontdimen2\font plus
\BIBentryALTinterwordstretchfactor\fontdimen3\font minus \fontdimen4\font\relax}
\providecommand{\BIBforeignlanguage}[2]{{%
\expandafter\ifx\csname l@#1\endcsname\relax
\typeout{** WARNING: IEEEtran.bst: No hyphenation pattern has been}%
\typeout{** loaded for the language `#1'. Using the pattern for}%
\typeout{** the default language instead.}%
\else
\language=\csname l@#1\endcsname
\fi
#2}}
\providecommand{\BIBdecl}{\relax}
\BIBdecl

\bibitem{roy2015approximate}
K.~Roy and A.~Raghunathan, ``Approximate computing: An energy-efficient computing technique for error resilient applications,'' in \emph{2015 IEEE Computer society annual symposium on VLSI}.\hskip 1em plus 0.5em minus 0.4em\relax IEEE, 2015, pp. 473--475.

\bibitem{alioto2018energy}
M.~Alioto, V.~De, and A.~Marongiu, ``Energy-quality scalable integrated circuits and systems: Continuing energy scaling in the twilight of moore’s law,'' \emph{IEEE Journal on Emerging and Selected Topics in Circuits and Systems}, vol.~8, no.~4, pp. 653--678, 2018.

\bibitem{haidous2022content}
A.~Haidous, W.~Oswald, H.~Das, and N.~Gong, ``Content-adaptable roi-aware video storage for power-quality scalable mobile streaming,'' \emph{IEEE Access}, vol.~10, pp. 26\,830--26\,848, 2022.

\bibitem{cao2018video}
Y.~Cao, Z.~Xu, P.~Qin, and T.~Jiang, ``Video processing on the edge for multimedia iot systems,'' \emph{arXiv preprint arXiv:1805.04837}, 2018.

\bibitem{young2025low}
L.~Young, D.~Wang, and N.~Gong, ``Low-cost wearable edge-ai device for diabetes management,'' in \emph{Proceedings of the Great Lakes Symposium on VLSI 2025}, 2025, pp. 238--244.

\bibitem{chen2015vcas}
D.~Chen, X.~Wang, J.~Wang, and N.~Gong, ``Vcas: Viewing context aware power-efficient mobile video embedded memory,'' in \emph{2015 28th IEEE International System-on-Chip Conference (SOCC)}.\hskip 1em plus 0.5em minus 0.4em\relax IEEE, 2015, pp. 333--338.

\bibitem{edstrom2016luminance}
J.~Edstrom, D.~Chen, J.~Wang, H.~Gu, E.~A. Vazquez, M.~E. McCourt, and N.~Gong, ``Luminance-adaptive smart video storage system,'' in \emph{2016 IEEE International Symposium on Circuits and Systems (ISCAS)}.\hskip 1em plus 0.5em minus 0.4em\relax IEEE, 2016, pp. 734--737.

\bibitem{chen2018viewer}
D.~Chen, J.~Edstrom, Y.~Gong, P.~Gao, L.~Yang, M.~E. McCourt, J.~Wang, and N.~Gong, ``Viewer-aware intelligent efficient mobile video embedded memory,'' \emph{IEEE Transactions on Very Large Scale Integration (VLSI) Systems}, vol.~26, no.~4, pp. 684--696, 2018.

\bibitem{edstrom2019content}
J.~Edstrom, Y.~Gong, A.~A. Haidous, B.~Humphrey, M.~E. Mccourt, Y.~Xu, J.~Wang, and N.~Gong, ``Content-adaptive memory for viewer-aware energy-quality scalable mobile video systems,'' \emph{IEEE Access}, vol.~7, pp. 47\,479--47\,493, 2019.

\bibitem{ROI-AwareVideoStorage}
A.~Haidous, W.~Oswald, H.~Das, and N.~Gong, ``Content-adaptable roi-aware video storage for power-quality scalable mobile streaming,'' \emph{IEEE Access}, vol.~10, pp. 26\,830--26\,848, 2022.

\bibitem{oswald2024flexible}
W.~Oswald, M.~S. Hossain, K.~Mooney, M.~Renteria-Pinon, M.~B. Hossain, M.~Shaban, J.~Wang, and N.~Gong, ``Flexible bit-truncation memory for low-power quality-adaptive video and deep learning storage,'' in \emph{2024 IEEE 15th International Green and Sustainable Computing Conference (IGSC)}.\hskip 1em plus 0.5em minus 0.4em\relax IEEE, 2024, pp. 87--92.

\bibitem{gong2025ai}
N.~Gong, J.~Wang, W.~Jin, H.~Das, and A.~Haidous, ``Ai-enabled efficient memory design for data-intensive applications,'' in \emph{AI-Enabled Electronic Circuit and System Design: From Ideation to Utilization}.\hskip 1em plus 0.5em minus 0.4em\relax Springer, 2025, pp. 83--108.

\bibitem{frustaci2016approximate}
F.~Frustaci, D.~Blaauw, D.~Sylvester, and M.~Alioto, ``Approximate srams with dynamic energy-quality management,'' \emph{IEEE Transactions on Very Large Scale Integration (VLSI) Systems}, vol.~24, no.~6, pp. 2128--2141, 2016.

\bibitem{model_compression_forDNN_survey}
Z.~Li, H.~Li, and L.~Meng, ``Model compression for deep neural networks: A survey,'' \emph{Computers}, vol.~12, no.~3, p.~60, 2023.

\bibitem{computational_complexity_rediction}
M.~B. Hossain, N.~Gong, and M.~Shaban, ``Computational complexity reduction techniques for deep neural networks: A survey,'' in \emph{2023 IEEE International Conference on Artificial Intelligence, Blockchain, and Internet of Things (AIBThings)}.\hskip 1em plus 0.5em minus 0.4em\relax IEEE, 2023, pp. 1--6.

\bibitem{gholami2022survey}
A.~Gholami, S.~Kim, Z.~Dong, Z.~Yao, M.~W. Mahoney, and K.~Keutzer, ``A survey of quantization methods for efficient neural network inference,'' in \emph{Low-Power Computer Vision}.\hskip 1em plus 0.5em minus 0.4em\relax Chapman and Hall/CRC, 2022, pp. 291--326.

\bibitem{VGG-16_original_paper}
K.~Simonyan and A.~Zisserman, ``Very deep convolutional networks for large-scale image recognition,'' \emph{arXiv preprint arXiv:1409.1556}, 2014.

\bibitem{liang2021pruning}
T.~Liang, J.~Glossner, L.~Wang, S.~Shi, and X.~Zhang, ``Pruning and quantization for deep neural network acceleration: A survey,'' \emph{Neurocomputing}, vol. 461, pp. 370--403, 2021.

\bibitem{he2023structured}
Y.~He and L.~Xiao, ``Structured pruning for deep convolutional neural networks: A survey,'' \emph{IEEE Transactions on Pattern Analysis and Machine Intelligence}, 2023.

\bibitem{hu2018squeeze}
J.~Hu, L.~Shen, and G.~Sun, ``Squeeze-and-excitation networks,'' in \emph{Proceedings of the IEEE conference on computer vision and pattern recognition}, 2018, pp. 7132--7141.

\bibitem{liu2021channel}
M.~Liu, W.~Fang, X.~Ma, W.~Xu, N.~Xiong, and Y.~Ding, ``Channel pruning guided by spatial and channel attention for dnns in intelligent edge computing,'' \emph{Applied Soft Computing}, vol. 110, p. 107636, 2021.

\bibitem{hu2022neural}
J.~Hu, Y.~Liu, and K.~Wu, ``Neural network pruning based on channel attention mechanism,'' \emph{Connection Science}, vol.~34, no.~1, pp. 2201--2218, 2022.

\bibitem{edstrom2017data}
J.~Edstrom, Y.~Gong, D.~Chen, J.~Wang, and N.~Gong, ``Data-driven intelligent efficient synaptic storage for deep learning,'' \emph{IEEE Transactions on Circuits and Systems II: Express Briefs}, vol.~64, no.~12, pp. 1412--1416, 2017.

\bibitem{IEEE_FloatingPointStandard}
``Ieee standard for floating-point arithmetic,'' \emph{IEEE Std 754-2019 (Revision of IEEE 754-2008)}, pp. 1--84, 2019.

\bibitem{sky130}
\BIBentryALTinterwordspacing
``Skywater sky130 pdk.'' [Online]. Available: \url{https://skywater-pdk.readthedocs.io/en/main/}
\BIBentrySTDinterwordspacing

\bibitem{HMDB_Dataset}
H.~Kuehne, H.~Jhuang, E.~Garrote, T.~Poggio, and T.~Serre, ``Hmdb: a large video database for human motion recognition,'' in \emph{2011 International conference on computer vision}.\hskip 1em plus 0.5em minus 0.4em\relax IEEE, 2011, pp. 2556--2563.

\bibitem{Git_Video_Emulator}
\BIBentryALTinterwordspacing
W.~Oswald, D.~Diggs, H.~Das, and A.~Haidous, ``Impact video truncation emulator.'' [Online]. Available: \url{https://github.com/LiamOswald/IMPACT_Video_Truncation_Emulator}
\BIBentrySTDinterwordspacing

\bibitem{Git_DNN_Emulator}
\BIBentryALTinterwordspacing
W.~Oswald and B.~Hossain, ``Impact dnn truncation emulator.'' [Online]. Available: \url{https://github.com/LiamOswald/IMPACT_DNN_Truncation_Emulator}
\BIBentrySTDinterwordspacing

\bibitem{zhang2015accelerating}
X.~Zhang, J.~Zou, K.~He, and J.~Sun, ``Accelerating very deep convolutional networks for classification and detection,'' \emph{IEEE transactions on pattern analysis and machine intelligence}, vol.~38, no.~10, pp. 1943--1955, 2015.

\bibitem{ResNet50_Paper}
\BIBentryALTinterwordspacing
K.~He, X.~Zhang, S.~Ren, and J.~Sun, ``Deep residual learning for image recognition,'' \emph{CoRR}, vol. abs/1512.03385, 2015. [Online]. Available: \url{http://arxiv.org/abs/1512.03385}
\BIBentrySTDinterwordspacing

\bibitem{krizhevsky2009learning}
A.~Krizhevsky, G.~Hinton \emph{et~al.}, ``Learning multiple layers of features from tiny images,'' 2009.

\bibitem{dufour2019finite}
J.-M. Dufour and J.~Neves, ``Finite-sample inference and nonstandard asymptotics with monte carlo tests and r,'' in \emph{Handbook of statistics}.\hskip 1em plus 0.5em minus 0.4em\relax Elsevier, 2019, vol.~41, pp. 3--31.

\bibitem{ren2015faster}
S.~Ren, K.~He, R.~Girshick, and J.~Sun, ``Faster r-cnn: Towards real-time object detection with region proposal networks,'' \emph{Advances in neural information processing systems}, vol.~28, 2015.

\bibitem{everingham2010pascal}
M.~Everingham, L.~Van~Gool, C.~K. Williams, J.~Winn, and A.~Zisserman, ``The pascal visual object classes (voc) challenge,'' \emph{International journal of computer vision}, vol.~88, pp. 303--338, 2010.

\bibitem{Xschem}
S.~Stefan, C.~Rafmag, and et. al, ``Xschem: Schematic capture and netlisting eda tool,'' \url{https://xschem.sourceforge.io/stefan/index.html}, accessed: 2023-10-26.

\bibitem{NGspice}
SourceForge, ``Ngspice: Geda (gpl’d suite and toolkit of electronic design automa- tion tools),'' \url{ ngspice.sourceforge.net/download.html}, accessed: 2024-10-26.

\bibitem{YoutubeM8}
S.~Abu-El-Haija, N.~Kothari, J.~Lee, P.~Natsev, G.~Toderici, B.~Varadarajan, and S.~Vijayanarasimhan, ``Youtube-8m: A large-scale video classification benchmark,'' \emph{arXiv preprint arXiv:1609.08675}, 2016.

\bibitem{tensorflow2015-whitepaper}
\BIBentryALTinterwordspacing
``{TensorFlow}: Large-scale machine learning on heterogeneous systems,'' 2015, software available from tensorflow.org. [Online]. Available: \url{https://www.tensorflow.org/}
\BIBentrySTDinterwordspacing

\end{thebibliography}
\bibliographystyle{IEEEtran}
\vspace{-0.5in}
\begin{IEEEbiography}
[{\includegraphics[width=1in,height=1.25in,clip]{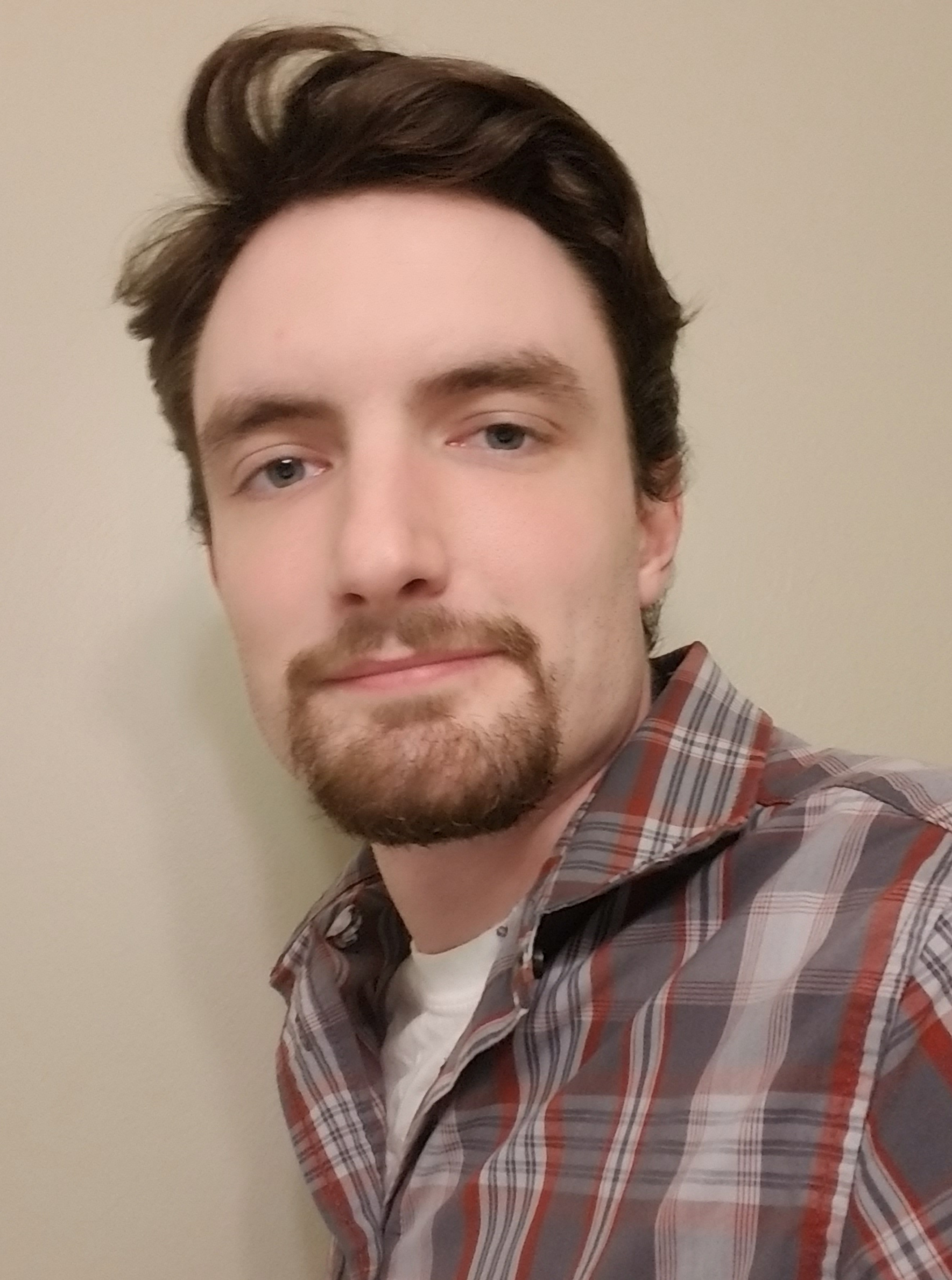}}]{William Oswald} Received the B.S. degree in computer engineering, a M.S. degree in electrical engineering, and a Ph.D. in systems engineering from the University of South Alabama in 2016, 2022, and 2024 respectively. He is currently employed as a senior machine learning engineer. His research interests include low-power circuits, image processing, autonomous navigation systems, physics-informed machine learning, and novel machine learning architectures. 
\end{IEEEbiography}

\vspace{-.3in}

\begin{IEEEbiography}[{\includegraphics[width=1in,height=1.25in,clip]{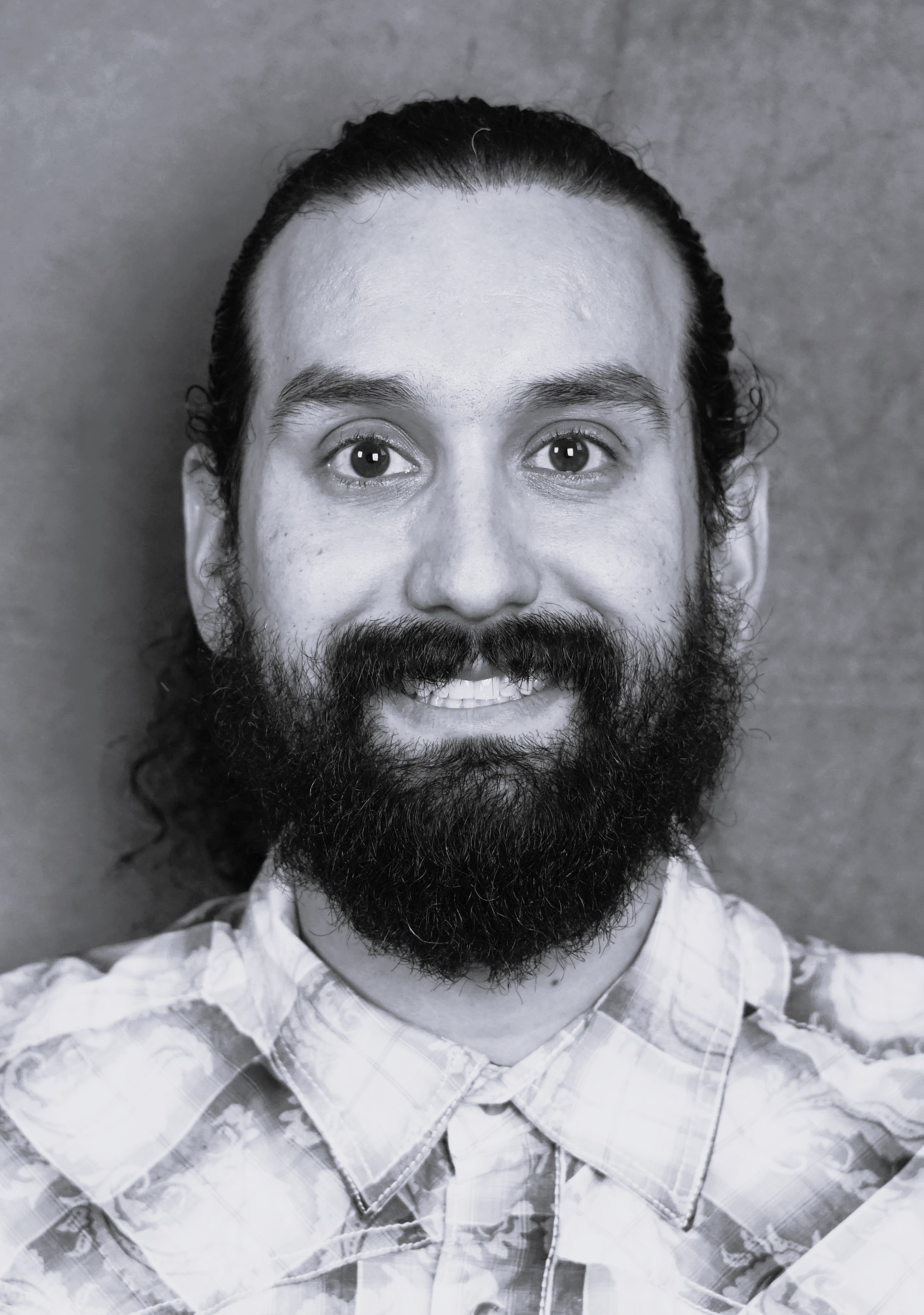}}]
{Mario Renteria-Pinon} (Member, IEEE) earned his B.S. degree in electrical and computer engineering from The University of Texas at El Paso in 2014, his M.S. degree in electrical engineering from the University of Washington in 2016, and his Ph.D. in engineering from New Mexico State University in 2023. He worked as a Postdoctoral Fellow for the IMPACT lab at the University of South Alabama, where he conducted research on AI hardware implementation, low-power memory, and privacy hardware. Currently, Dr. Renteria-Pinon is an Assistant Professor at New Mexico State University. His research interests extend to analog and mixed-signal integrated circuits, low-power sensors, data converters, and ASICs for biomedical applications.
\end{IEEEbiography}

\vspace{-.3in}

\begin{IEEEbiography}[{\includegraphics[width=1in,height=1.25in,clip]{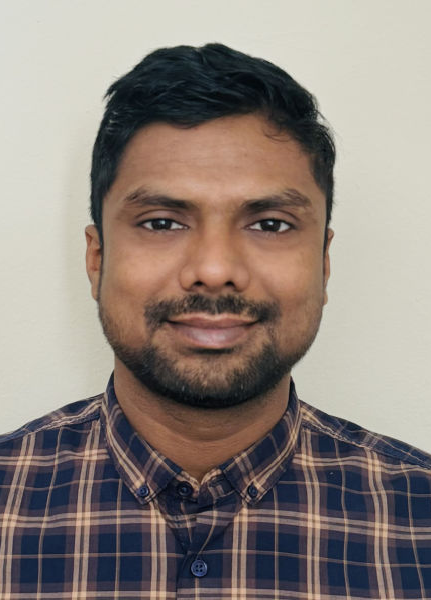}}]
{Md. Sajjad Hossain} received his B.Sc. degree in Electronics and Telecommunication Engineering from Rajshahi Engineering and Technology, Bangladesh in 2017, his M.S. degree in IT Convergence Engineering from Kumoh National Institute of Technology, South Korea in 2021 and is currently pursuing a Ph.D. in Electrical Engineering from the University of Alabama. His research interests include memory design, edge computing, IoT and machine learning.
\end{IEEEbiography}

\vspace{-.3in}

\begin{IEEEbiography}[{\includegraphics[width=1\linewidth,height=1.25\linewidth,clip]{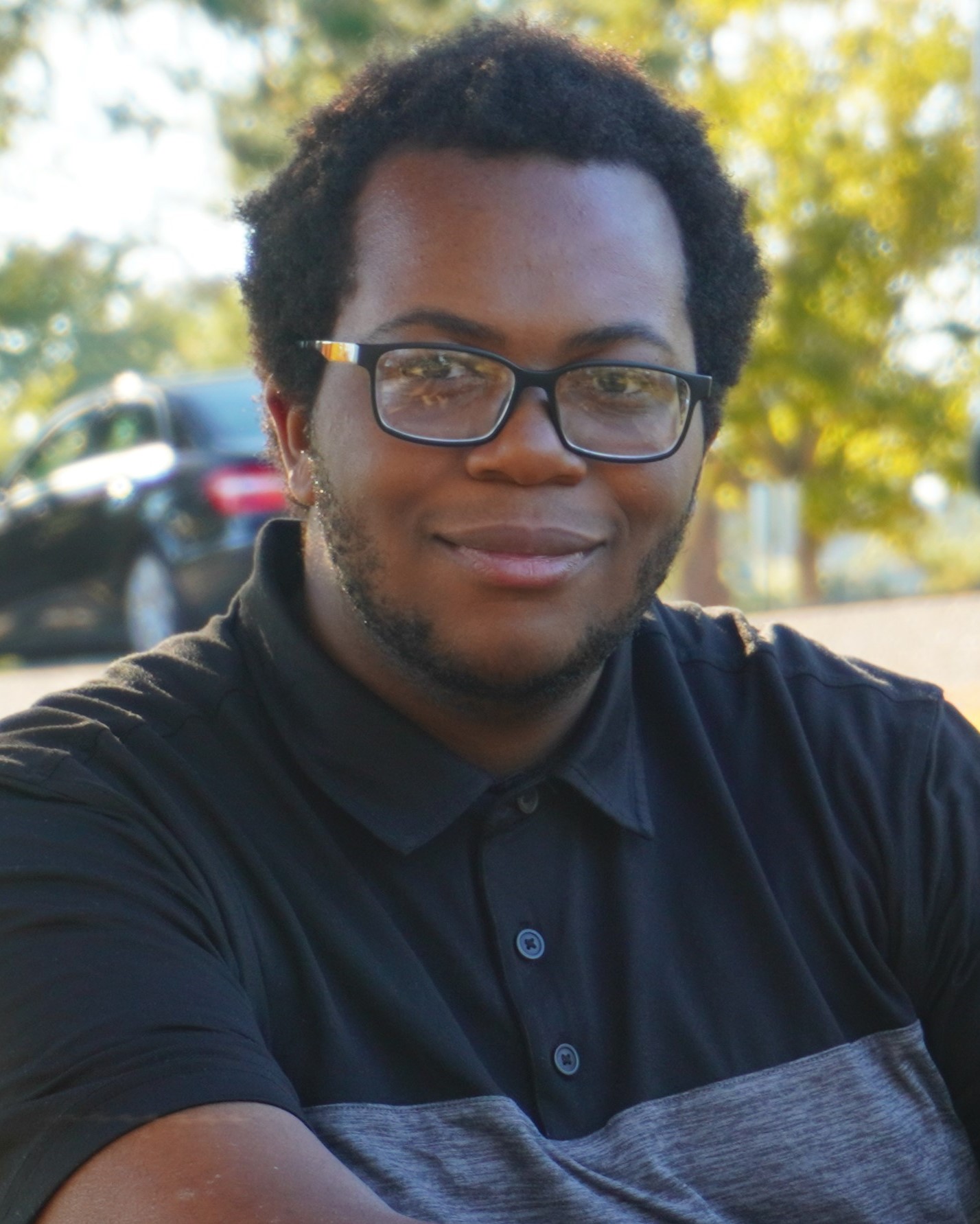}}]{Kyle Mooney}  (Student Member, IEEE) received his B.S. degree in computer engineering from the University of South Alabama in 2023 and is currently pursuing a Ph.D. in Electrical Engineering from the University of Alabama. His research interests include memory design, edge computing, and artificial intelligence.
\end{IEEEbiography}

\vspace{-.3in}

\begin{IEEEbiography}[{\includegraphics[width=1in,height=1.25in,clip]{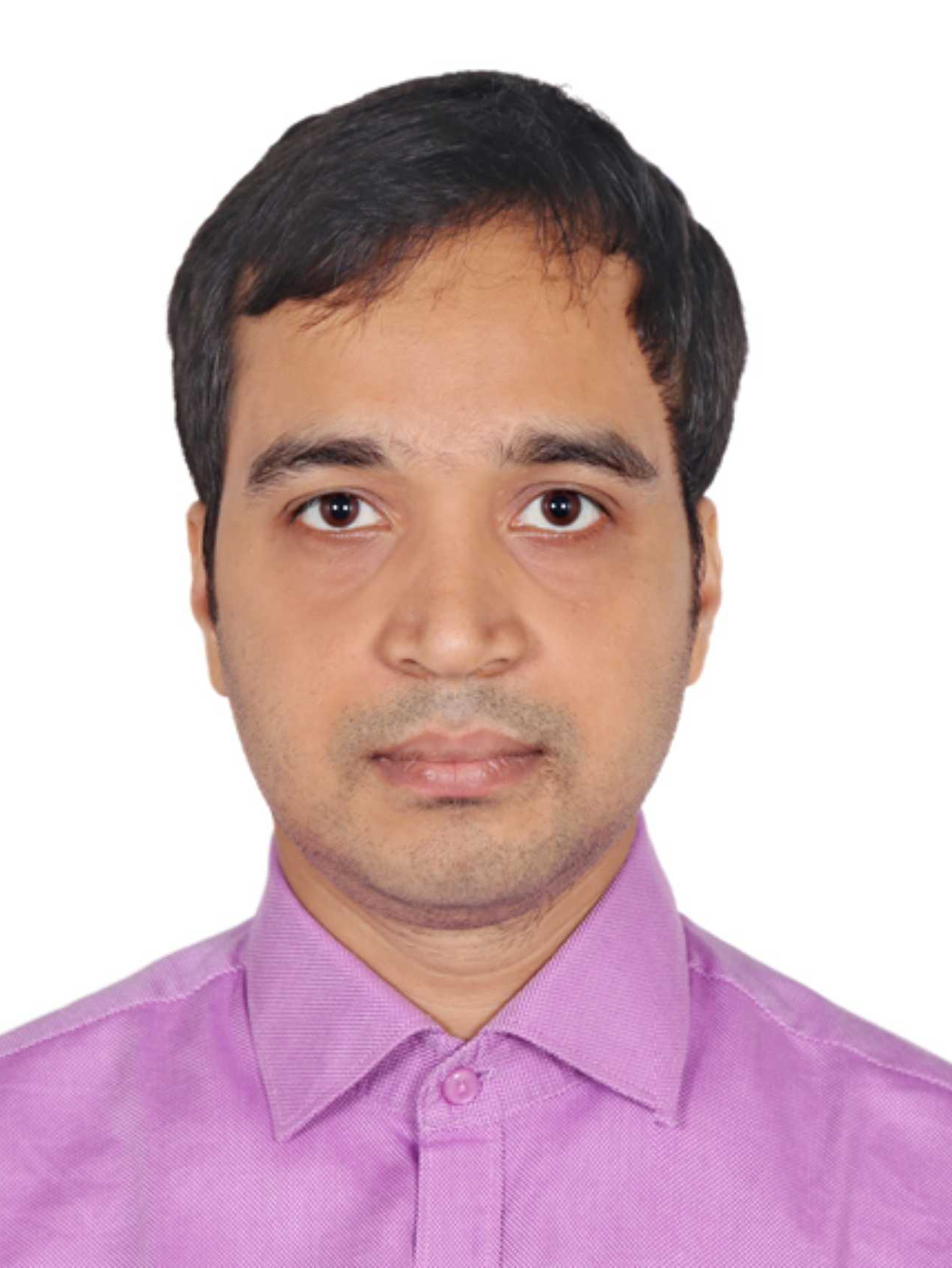}}]
{Md. Bipul Hossain} is pursuing the Ph.D. in Systems Engineering (Electrical and Computer Engineering Track) at the University of South Alabama. He has received the M.S. and B.S. degrees in Information and Communication Engineering from the Islamic University, Bangladesh . In addition, he served as a faculty member at Noakhali Science and Technology University during the period from Sept. 2018 until Dec. 2022. He is currently investigating and introducing novel deep learning model optimization algorithms for low-power low-complexity edge devices.  
\end{IEEEbiography}

\vspace{-.3in}

\begin{IEEEbiography}[{\includegraphics[width=1in,height=1.25in,clip]{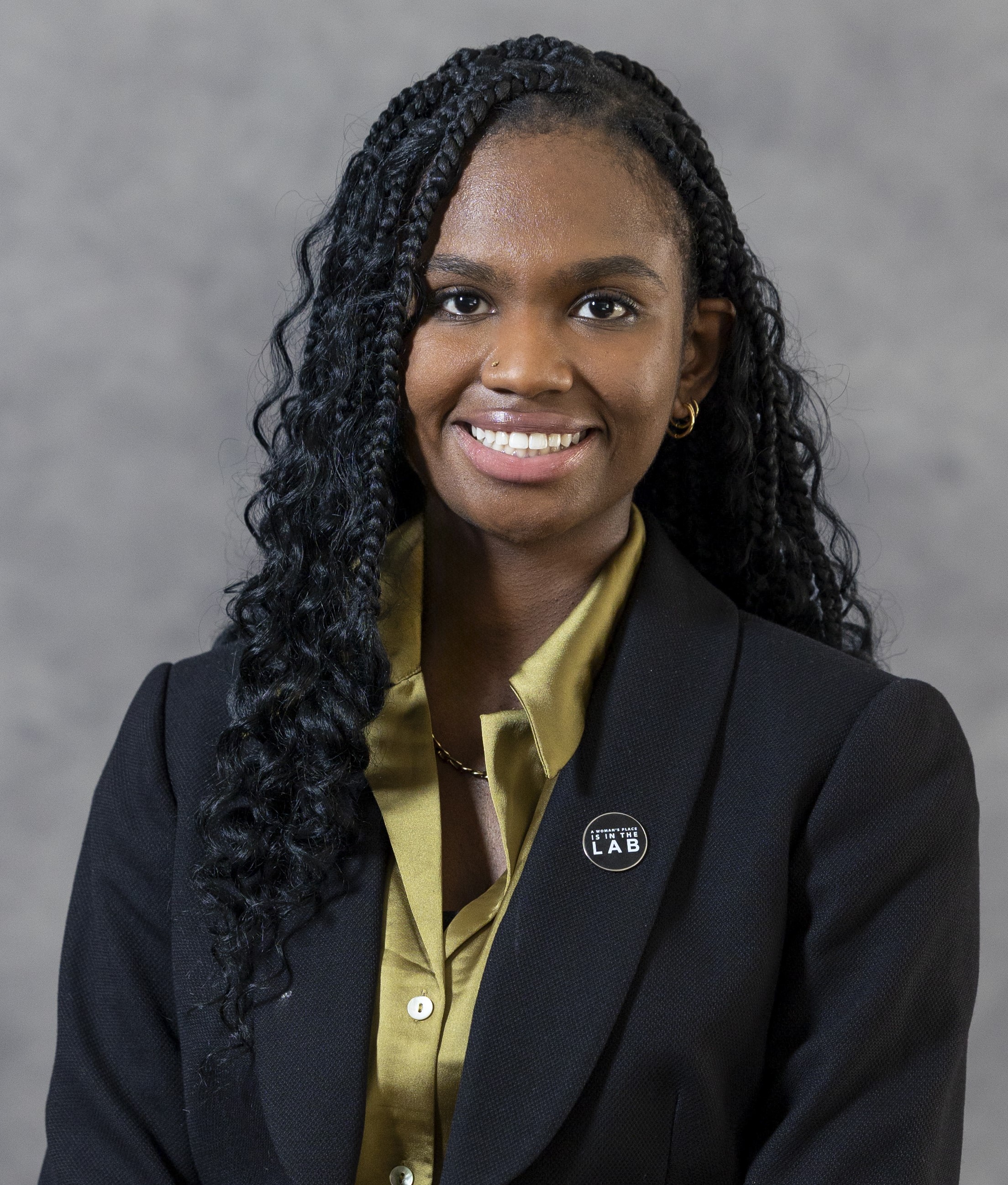}}]
{Destinie Diggs} is currently pursuing the MS degree in Electrical Engineering at the University of Alabama. She received her B.S. degree in computer engineering from the University of South Alabama in 2025. Her research interests include video systems and Edge AI.
\end{IEEEbiography}

\vspace{-.3in}

\begin{IEEEbiography}[{\includegraphics[width=1in,height=1.25in,clip]{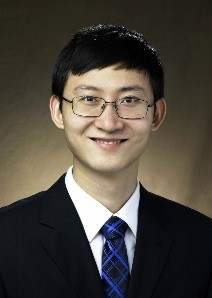}}]
{Yiwen Xu} received his
Ph.D. degree in systems and industrial engineering from the University of Arizona. He is currently an algorithmic trader and small business owner. His interests are machine learning, algorithmic trading, AI-implemented web design, and applied operations research.
\end{IEEEbiography}

\vspace{-.3in}

\begin{IEEEbiography}[{\includegraphics[width=1in,height=1.25in,clip]{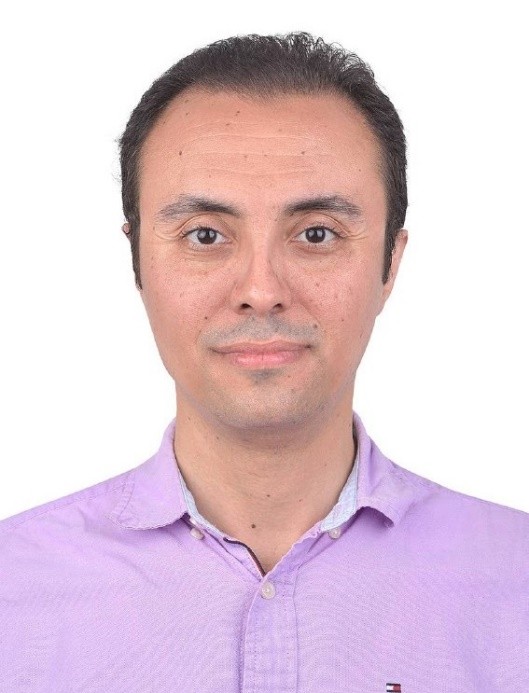}}]
{Mohamed Shaban} has been an assistant professor in the Electrical, and Computer Engineering department at the University of South Alabama since January 2020. He has received the Ph.D., and M.S. degrees in Computer Engineering from the University of Louisiana at Lafayette in 2016, and 2012 respectively. He has also received the M.S. degree in Electrical Communications Engineering, and the B.S. degree (Excellent with Honors) in Electronics, and Communications Engineering from Mansoura University, Egypt in 2010, and 2006 respectively. He has previously served as an Assistant Professor of Computer Science at Southern Arkansas University, Graduate Teaching Assistant at the University of Louisiana at Lafayette and an Assistant Lecturer at Mansoura University. Dr. Shaban is a senior member IEEE, and a senior member of the IEEE signal processing and IEEE engineering in medicine and biology societies. His current research interests are in the fields of Signal, and Image Processing for Biomedical Applications, Machine, and Deep Learning Applications, Edge Artificial Intelligence Optimization and Time-Series Analysis.
\end{IEEEbiography}

\vspace{-.3in}

\begin{IEEEbiography}[{\includegraphics[width=1\linewidth,height=1.25\linewidth,clip]{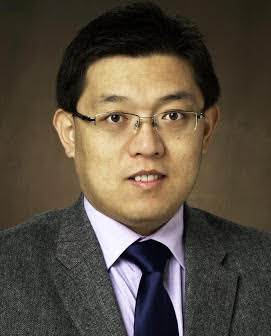}}]{Jinhui Wang} (Senior Member, IEEE) is currently a Full Professor and Larry Drummond Endowed Chair with the Department of Electrical and Computer Engineering at the University of Alabama, Tuscaloosa, AL, USA. His research interests include: (1) VLSI System, Digital and Mixed-Signal Integrated Circuit (IC) Design, 3D and 2.5D IC Design, and Emerging Memory; (2) AI Hardware Design, Post/Beyond CMOS Device, such as Memristors, Based Neuromorphic Computing System; and (3) Post/Beyond CMOS Devices Enabled Cybersecurity and Internet of Things (IoT) Systems. He has published over 200 refereed journal/conference papers and book chapters as well as 31 patents in the area of emerging semiconductor technologies. His previous work has received the Best Paper Award/Nomination at DATE 2021, ISVLSI 2019, ISLPED 2016, ISQED 2016, and EIT 2016.
\end{IEEEbiography}

\vspace{-.3in}

\begin{IEEEbiography}[{\includegraphics[width=1\linewidth,height=1.25\linewidth,clip]{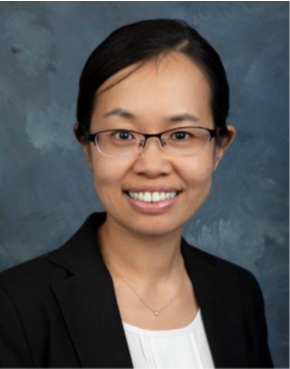}}]{Na Gong (M’13)}  received the Ph.D. degree in computer science and engineering from the State University of New York, Buffalo, in 2013. Currently, Dr. Gong is a professor in the Department of Electrical and Computer Engineering at the University of Alabama. Her research interests include power-efficient computing circuits and systems, memory optimization, AI hardware, and hardware privacy. She is the recipient of the best paper nomination from ISVLSI’19, best paper award from EIT’16, best paper nominations from ISQED’16 and ISLPED’16.
\end{IEEEbiography}


 

\vfill

\end{document}